\documentclass[3p]{elsarticle}

\usepackage{setspace}
\journal{Journal of \LaTeX\ Templates}
\biboptions{numbers,sort&compress}

\usepackage{graphics} % for pdf, bitmapped graphics files
\usepackage{graphicx}
\usepackage{epsfig} % for postscript graphics files
\usepackage{times} % assumes new font selection scheme installed
\usepackage{amsmath} % assumes amsmath package installed
\usepackage{amssymb}  % assumes amsmath package installed
\usepackage{bm}
\usepackage{bbm}
\usepackage{mathtools}
\usepackage{nccmath}
\usepackage{color,url}
\usepackage{neuralnetwork}
\usepackage{xpatch}
\usepackage{bbm}
\usepackage[hidelinks]{hyperref}

\usepackage {tikz}
\usetikzlibrary {positioning}
\definecolor {processblue}{cmyk}{0.96,0,0,0}
\usetikzlibrary{shapes,arrows}
\usetikzlibrary{calc}
\usetikzlibrary{matrix,chains,decorations.pathreplacing}

\newcommand{\trp}{\intercal}

\DeclareMathOperator*{\minimize}{min\,}

\newtheorem{theorem}{Theorem}

\newdefinition{example}{Example}
\newdefinition{remark}{Remark}
\newdefinition{definition}{Definition}
\newproof{proof}{Proof}
\newproof{pot}{Proof of Theorem \ref{thm2}}

\usepackage[ruled,linesnumbered]{algorithm2e}
\usepackage{algpseudocode}

\begin{document}
\begin{frontmatter}
	
	\title{Model-Free Design of Control Systems over Wireless Fading Channels\tnoteref{mytitlenote}}
	\tnotetext[mytitlenote]{Preliminary results have been presented at the 2020 IEEE International Workshop on Signal Processing Advances in Wireless Communications (SPAWC) \cite{silva2019optimal} and at the 2020 IEEE Conference on Decision and Control (CDC) \cite{lima_cdc_2020_constrained}. We extend the preliminary results presented in the conference papers to include the design of communication-aware control policies, present theoretical results on the stability of near-optimal policies, and present new numerical experiments.}

	\author[upenn]{Vinicius Lima\corref{mycorrespondingauthor}}
	\cortext[mycorrespondingauthor]{Corresponding author}
	\ead{vlima@seas.upenn.edu}
	
	\author[intel]{Mark Eisen}
	\ead{mark.eisen@intel.com}
	
	\author[oxford]{Konstantinos Gatsis}
	\ead{konstantinos.gatsis@eng.ox.ac.uk}
	
	\author[upenn]{Alejandro Ribeiro}
	\ead{aribeiro@seas.upenn.edu}
	
	\address[upenn]{University of Pennsylvania, Philadelphia, PA 19104 USA}
	\address[intel]{Intel Corporation, Hillsboro, OR}
	\address[oxford]{Department of Engineering Science, University of Oxford, Parks Road, Oxford, OX1 3PJ, UK}
	
	\begin{abstract}
		Wireless control systems rely on wireless networks to exchange information between spatially distributed actuators, plants and sensors. The noise in wireless channels renders traditional control policies suboptimal, and their performance is moreover directly dependent on the allocation of wireless resources between control loops sharing the wireless medium, making the design of control and resource allocation policies critical to achieve reliable performance of the wireless control system. The co-design of control-aware resource allocation policies and communication-aware controllers, however, is a challenging problem due to its infinite dimensionality, existence of constraints on the use of communication resources or performance of individual plants, and need for explicit knowledge of the plants and wireless network models. 
To overcome those challenges, we propose a constrained reinforcement learning approach to design model-free control-aware communication policies and communication-aware control policies for wireless control systems.	
		We demonstrate the near optimality of control system performance and stability using near-universal policy parametrizations and present a practical model-free algorithm to learn the co-design policy. Numerical experiments show the strong performance of learned policies over baseline solutions.
	\end{abstract}
	
	
\end{frontmatter}

%%%%%%%%%%%%%%%%%%%%%%%%%%%%%%%
%%% SECTION : Introduction  %%%
%%%%%%%%%%%%%%%%%%%%%%%%%%%%%%%
\section{Introduction}
The use of wireless networks to exchange information between actuators, plants and sensors in control systems can add significant flexibility and cost savings to the deployment, installation and maintenance of control systems \cite{park_wireless_2018}.
However, the replacement of wired communications with wireless makes the careful design of control and communication policies to keep the control system operational a critical and challenging problem \cite{hespanha_survey_2007, park_wireless_2018}.
Wireless networks  are characterized by rapidly changing channel transmission conditions known as \emph{fading} \cite{tse_2005, schenato_foundations_2007}; they are also, in general, noisier than standard wired communication and subject to packet losses as the fading conditions deteriorate. 
That implies that components of a wireless control system (WCS) might have to occasionally operate under noisy or missing information --- whereas traditional closed loop control systems are usually designed under the assumption that communication is fast and reliable \cite{park_wireless_2018}. 
Designing wireless control systems involves then finding an optimal way to allocate the resources available in the network among the plants sharing that communication medium, as well as devising control policies that are able to maintain plants operating reliably in face of eventual information loss. Moreover, a robust design of communication and control should allow for the addition of  constraints to either wireless resources or control operation. Finding an optimal solution to this constrained co-design problem, however, is often intractable; meanwhile, finding approximate solutions inevitably requires knowledge of the underlying dynamics of the controlled plants and communication network, which might be inaccurate or often unavailable in practice. 

In standard wireless networks,  resource allocation involves optimizing communication performance measures such as throughput, resource consumption,  and latency against stochastic noise and wireless fading  in the communication channel\cite{yaxin_li_scheduling_2001,fattah_leung_scheduling_wl_2002,Eryilmaz_srikant_2007,Ribeiro2012}. The resulting problem  consists in  optimizing some performance measure over an allocation function, which leads to an infinite dimensional optimization problem that is often hard to solve. 
Resource allocation problems, however, can be cast as statistical learning problems \cite{Ribeiro2012}, motivating the use of data-driven approaches to overcome the challenging nature of the problem and design allocation policies in wireless networks  \cite{sun_learning_2017, eisen_learning_2018, liang_opc_dnns, liang2019deep, zappone2019wireless}. 

To enable the remote operation of \emph{autonomous systems} over wireless networks, the design of resource allocation and control policies should explicitly take into account the interplay between network resources and the dynamics of the plants sharing the network. 
Resource allocation and scheduling for control systems have been studied extensively, see \cite{rehbinder_scheduling_2004, mo_sensor_2011, shi_optimal_2011, gatsis_opportunistic_2015, charalambous_resource_2017, eisen2019control, wu_optimal_2018} to name a few. Meanwhile, the \emph{co-design} of resource allocation and control policies for linear WCSs is studied in \cite{molin2009lqg, bao_iterativedesign_2011, ramesh_schedulers_2013, gatsis_optimal_2014, rabi_sepdesign_2016}. 
By introducing a decentralized information structure \cite[Proposition 1]{gatsis_optimal_2014}, for example, the joint optimization problem can be decoupled into separate control and allocation or scheduling problems, making the co-design problem more tractable: the optimal control policy is a standard LQR controller, and near-optimal communication policies can be found via approximate dynamic programming. Separating the design of control, scheduling and estimation policies, under the appropriate assumptions, is also the approach taken in \cite{molin2009lqg, bao_iterativedesign_2011, ramesh_schedulers_2013, rabi_sepdesign_2016}. The design of control and triggering policies for networked, possibly nonlinear systems is studied in \cite{heydari_codesign_adp}, whereas the design of control and power transmission policies for possibly nonlinear systems is approached in  \cite{hu_codesign_fan_2019}. Although communication costs are penalized, the proposed approach does not take resource constraints into account.

The inaccuracy or unavailability of models in practice, however, limits the applicability of the approaches mentioned above, especially in the presence of wireless fading effects,  which leads to rapid changes in network performance and impacts the performance of the control policy.
Recent advances in machine learning, on the other hand, have motivated the search for data-driven approaches --- notably reinforcement learning (RL) --- to aid the design of resource allocation and scheduling policies in WCSs \cite{demirel_deepcas:_2018,leong_deep_2018,baumann_drl_etc, redder2019deep, silva2019optimal}.
Model-free value-based RL methods are used to design scheduling policies for WCSs in \cite{demirel_deepcas:_2018, leong_deep_2018}. A learned scheduling policy is combined with a model-based controller optimized for the wireless scheduler in \cite{redder2019deep}, while the use of actor-critic algorithms to co-design communication and control policies in a simple event-triggered communication model is proposed in \cite{baumann_drl_etc}. 
A thorough design of control and communication policies in WCSs, however, should necessarily include
 (i) wireless fading states in the communication/control design and (ii) allow for the inclusion of hard long-term constraints on the performance of the system in terms of, e.g., resource utilization or plant safety constraints.

In this paper, we then discuss the model-free co-design of control and resource allocation policies in wireless control systems over fading channels with limited, centralized network access, and subject to additional long-term constraints. The wireless channel is noisy and subject to packet loss based on a channel state and the resource allocated to that particular signal, and thus proper allocation of limited network resources as well as the design of control policies is fundamental to achieve good performance. The co-design problem is formulated as finding optimal resource allocation and control policies, using plant and channel states as inputs, that jointly optimize system performance under further system constraints (Section \ref{sec:prob_formulation}). 
This however results in a hard optimization problem that is intractable to solve exactly and fundamentally relies on precise model knowledge of control system dynamics and communication models. 

The novelty of the paper lies in casting the joint design problem in WCSs  as a constrained reinforcement learning problem (Section \ref{sec:resource_RL}).
Using near-universal parametrizations to parametrize resource allocation and control policies, we rely on recent results from the literature on constrained reinforcement learning to bound the gap between the optimal solution of the codesign problem and the solution found by a near-universal parametrization (Section \ref{sec:resource_RL}).
We also show that, if the optimal solution of the codesign problem is sufficiently stable, so will be the solution found by a near-universal parametrization (Section \ref{sec:resource_RL}).
We further present a primal-dual learning algorithm to find codesign policies that alternates between updating the model parameters via reinforcement learning iterations and updating a dual variable that enforces constraint satisfaction, all without explicit knowledge of plant dynamics or communication models (Section \ref{sec_learning}).  Numerical experiments (Section \ref{sec:num_exp}) show the strong performance of such policies over baseline allocation solutions, highlighting the benefits the proposed approach brings to the design of control systems over wireless fading channels. Throughout the paper, uppercase letters refer to matrices and lowercase letters to vectors. Positive (semi)definiteness of a matrix is indicated by $X (\geq) > 0$. 

%%%%%%%%%%%%%%%%%%%%%%%%%%%%%%%
%%% SECTION : Wireless Control Systems  %%%
%%%%%%%%%%%%%%%%%%%%%%%%%%%%%%%
\section{Wireless Control Systems}
\label{sec:prob_formulation}
Here we consider a collection of $m$ independent plants communicating over a common wireless network as shown in Figure \ref{fig:wl_control_sys}.
At each time instant $t$, plants send their current state vector $x_t^{(i)} \in \mathbb{R}^p, i = 1, \dots, m, $ to an edge device containing a shared wireless access point (AP) and a centralized, remote controller (RC).
Based on that information, the RC computes the corresponding control actions and sends the control signals back to the plants. The dynamics of each plant $i$ is given by a discrete, time-invariant function $f^{(i)}: \mathbb{R}^{p} \times \mathbb{R}^q \to \mathbb{R}^p $ mapping  a current state vector $x_t^{(i)}$ and corresponding control input $u_t^{(i)} \in \mathbb{R}^q$ to the next state of the system. Each plant is further affected by some random i.i.d. noise $w^{(i)}_t \in \mathbb{R}^p$ with mean zero and covariance matrix $W\in \mathbb{R}^p$, $W > 0$,  standing for eventual disturbances, leading to
\begin{equation}
x^{(i)}_{t + 1} = f^{(i)}(x_t^{(i)}, u_t^{(i)}) + w^{(i)}_t, i = 1, \dots, m.
\label{eq:ind_plant_dyn}
\end{equation}

Note that in \eqref{eq:ind_plant_dyn} the control signal $u^{(i)}$ is computed remotely and then sent to the plants over a wireless network. In standard control systems, one usually assumes perfect communication between plant and controller, and the control policy is designed to optimize a performance criterion involving the plant states and control actions, as in the classical linear quadratic regulator problem. That assumption does not hold in this setting, however. 
Operating control systems over a wireless channel is made complicated by the fact that the wireless communication medium over which the control loop is closed is inherently noisy and the network is resource limited. This, in turn, causes the plant to occasionally operate in open loop due to either packet loss or withheld transmission. If the transmission of the control signal is successful, the feedback control loop is closed, and the plant executes the correct control action as instructed by the RC. When the plant cannot reliably receive the signal, however, we assume it does not execute any control action. 
Under this model, the control input $u^{(i)}_t$ is governed by 
 \begin{equation}
 u^{(i)}_{t} = 
 \begin{cases}
    g^{(i)}(x_t^{(i)}), \text{ closed loop},  \\
    0, \qquad \text{ open loop}, 
 \end{cases}
 \label{eq:plant_switched}
\end{equation}
with $g^{(i)}: \mathbb{R}^p \to \mathbb{R}^q$ a \emph{remote} control policy implemented at the RC using current state information. 

\begin{remark} Note that in \eqref{eq:plant_switched} we can consider a more general formulation with the switched dynamics
 \begin{equation}
 u^{(i)}_{t} = 
 \begin{cases}
    g^{(i)}(x_t^{(i)}), \text{ closed loop},  \\
    \tilde{g}^{(i)}(u_{t-1}^{(i)}), \qquad \text{ open loop},
 \end{cases}
 \label{eq:plant_switched_general}
\end{equation}
where we further consider a \emph{local} control policy  $\tilde{g}^{(i)}: \mathbb{R}^q \to \mathbb{R}^q$ that can be implemented by the plant using only previous control information when the feedback loop is open.  Simple cases of the local policy include, e.g., using the previous input, i.e. $\tilde{g}^{(i)}(u_{t}^{(i)}) = u_{t-1}^{(i)}$, or no input, i.e.  $\tilde{g}^{(i)}(u_{t}^{(i)}) = 0$ as in \eqref{eq:plant_switched}. Numerical experiments in this paper consider the latter case, and thus we adopt formulation \eqref{eq:plant_switched} in the remainder of the paper for the sake of clarity. 
In principle, one could design the local control policy as well, but that results in an (even) harder optimization problem that we do not address here. We refer the interested reader to \cite{gatsis_optimal_2014} for a more detailed discussion.
\end{remark}

\begin{remark}
Observe in the Fig. \ref{fig:wl_control_sys} and the switched dynamics in \eqref{eq:plant_switched} that the open-loop configuration of the wireless control system is restricted to the actuation, or ``downlink'', stage of the control cycle and not the sensing, or ``uplink'', stage of the cycle. This is to say that we assume that the AP always has state information of all plants available and can thus allocate resources and compute control decisions, but control packets may be lost. This model is reasonable in practice as sensing devices, e.g. cameras, are more likely to be stationary in the environment and require high data rate transmissions, thus motivating the use of wired connections that do not suffer packet loss. Alternatively, the possible mobility of the plants themselves necessitates a wireless connection to the AP, which is practically feasible due to lower data rates needed to transmit control signals. In any case, we point out that the methodology developed in this work can be easily extended to the case of wireless uplink.
\end{remark}

%%%%%%%%%%%%%%%%%%%%%%%%%%%
%%%%%%% FIGURE : WCS %%%%%%
%%%%%%%%%%%%%%%%%%%%%%%%%%%
\tikzstyle{block} = [draw, fill=blue!40, rectangle, 
    minimum height=1cm, minimum width=5em]
\tikzstyle{largerblock} = [draw, fill=blue!20, rectangle, 
    minimum height=1cm, minimum width=21em,style=dashed]
    \tikzstyle{largerblock2} = [draw, fill=blue!40, rectangle, 
    minimum height=1cm, minimum width=16em]
\tikzstyle{sum} = [draw, fill=blue!20, circle, node distance=1cm]
\tikzstyle{input} = [coordinate]
\tikzstyle{output} = [coordinate]
\tikzstyle{pinstyle} = [pin edge={to-,thin,black}]

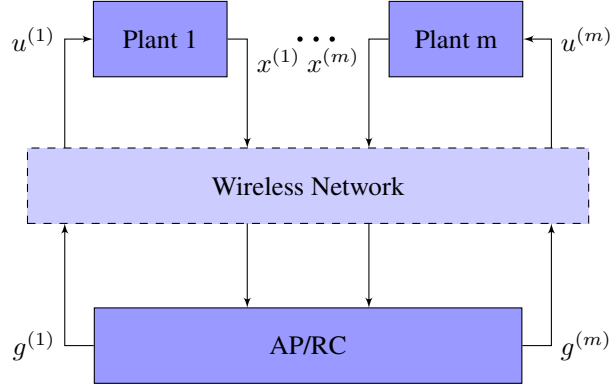
\begin{figure}
 \centering
\begin{tikzpicture}[auto, node distance=2cm,>=latex']
    \node [largerblock] (measurements) {Wireless Network};
        \node [block, above left of=measurements,
            node distance=2.75cm] (system) {Plant 1};
            \node [block, above right of=measurements,
            node distance=2.75cm] (systemm) {Plant m};
    \node [largerblock2, below of=measurements,node distance = 2.12cm](central_controller){AP/RC};
    % plants send states ...
    \node [output, above of=measurements, node distance = 0.5cm] (aux1) {};
    \node [output, right of=aux1, node distance = 0.8cm] (teste) {};
     \node [input, right of=teste, node distance = 2.7cm] (teste2) {};
     \node [input, left of=aux1, node distance = 0.8cm] (auxsys1) {};
     \node [input, left of=auxsys1, node distance = 2.7cm] (auxsys1-1) {};
    \draw [->] (system.east) -| node [name=x, below right] {$x^{(1)}$}(auxsys1);
    \draw [->] (systemm.west) -| node [name=xm, below left] {$x^{(m)}$}(teste);
    \foreach \i in {0.45,0.55,0.65}
        \fill ($(system.east)!\i!(systemm.west)$) circle[radius=1pt];
    \node [output,below of=measurements, node distance = 0.5cm](aux_below){};
    \node [output, right of=aux_below, node distance = 0.8cm](aux_below_right){};
    \node [output, left of=aux_below, node distance = 0.8cm](aux_below_left){};
    \node [output,above of=central_controller, node distance = 0.5cm](aux_central){};
    \node [output, right of=aux_central, node distance = 0.8cm](aux_central_right){};
    \node [output, left of=aux_central, node distance = 0.8cm](aux_central_left){};
    \draw [->] (aux_below_right) -- (aux_central_right);
    \draw [->] (aux_below_left) -- (aux_central_left);
    \node [input, left of=aux_below_left, node distance = 2.4cm] (aux_ap1){};
    \node [input, left of=auxsys1, node distance = 2.4cm] (aux_ap3){};
\node [input, right of=aux_below_right, node distance = 2.4cm] (aux_ap2){};    
\node [input, right of=teste, node distance = 2.4cm] (aux_ap4){};   
    \draw [->] [->] (aux_ap3) |- node [name=resource1, left] {$u^{(1)}$}(system.west);
    \draw [->] [->] (aux_ap4) |- node [name=resourcem, right] {$u^{(m)}$}(systemm.east);
     \draw [->]  (central_controller.west) -| node [name=pow1, left] {$g^{(1)}$} (aux_ap1);
    \draw [->] (central_controller.east) -| node [name=powm, right] {$g^{(m)}$ } (aux_ap2);
\end{tikzpicture}
\caption{ Wireless control system made up by a collection of $m$ independent plants with internal states $x^{(i)}, i = 1, \dots, m$. Plants communicate with a remote, or edge, controller (RC) over a wireless communication network. The wireless network consists of different channels with wireless fading states $h^{(i)}$. Access to the network is managed by an access point (AP) co-located with the RC.
} 
\label{fig:wl_control_sys} 
\end{figure}
%%%%%%%%%%%%%%%%%%%%%%%%%%%
%%%%%%%%%%%%%%%%%%%%%%%%%%%
%%%%%%%%%%%%%%%%%%%%%%%%%%%

%%%%%%%%%%%%%%%%%%%%%%%%%%%
%%%% SUBSECTION : Wireless Comm Model %%%%
%%%%%%%%%%%%%%%%%%%%%%%%%%%
\subsection{Wireless communication model}
\label{subsec:wireless_comm_model}
Wireless communication channels are prone to packet loss due to random disturbances present in the medium. Moreover, wireless channels are characterized by rapidly changing transmissions strengths, known as wireless fading \cite[ch. 2]{tse_2005}. 
The current wireless fading in the channel and the resource level with which an information packet is sent will in turn impact the reliability of that communication channel.
 Let $h^{(i)} \in \mathcal{H} \subseteq \mathbb{R}^{n}_+$  be a random variable drawn from a probability distribution $\chi(h)$ representing the transmission power of the transmitter associated to plant $i$.  Let also $\alpha^{(i)} \in \mathbb{R}^{n}_+$ the resource allocated to the signal sent by plant $i$. Given the communication model, current channel states $h := [h^{(1)}, \hdots, h^{(m)}]$, and set of allocated resources $\alpha := [\alpha^{(1)}, \hdots, \alpha^{(m)}]$, each plant experiences a signal-to-noise-ratio (SNR) given by a function $\varsigma^{(i)}(h, \alpha) : \mathbb{R}^{m \times n}_{+} \times \mathbb{R}^{m \times n}_+ \rightarrow \mathbb{R}_+$, with 
\begin{equation}
\varsigma^{(i)}(h, \alpha) := h^{(i)} \alpha^{(i)}.
\end{equation}

The SNR value experienced by each plant will determine the probability of successfully receiving the control signal. 
Let then $v: \mathbb{R}_{+} \rightarrow [0,1]$ a function that, given an SNR value, returns the probability of successful transmission. Under SNR $\varsigma$ a particular communication channel is limited by its capacity, and a packet can almost surely be successfully decoded so long as the fixed transmission rate $r$ does not exceed the channel capacity, otherwise it is almost surely lost \cite[ch. 5]{tse_2005}. In practice, fixed packet sizes lead to a delivery rate function $v(\varsigma)$ that takes a continuum of values in $[0, 1]$. 

In the context of wireless control systems, $v(\varsigma)$ gives the probability of closing the control loop at time instant $t$ under some resource allocation $\alpha_t$ and channel state $h_t$. The controller dynamics in \eqref{eq:plant_switched} can then be written as
\begin{equation}
 u^{(i)}_{t} = 
 \begin{cases}
    g^{(i)}(x_t^{(i)}), \text{ w.p. } v( \varsigma^{(i)}(h_t, \alpha_t)),  \\
    0, \qquad \text{ w.p. } 1- v(\varsigma^{(i)}(h_t, \alpha_t)). 
 \end{cases}
 \label{eq:plant_switched_prob}
\end{equation}
Note that, according to this model, it is possible to use the amount of resource assigned to a particular control signal to regulate the reliability of that packet transmission. 
As can be seen in \eqref{eq:plant_switched_prob}, allocating more resource to the control signal sent to a particular plant will increase the probability of that control loop closing, and, in turn, of the plant executing the correct control action. Most practical systems, however, have limited resources to be distributed between the communication channels, and properly distributing communication resources among the plants is thus essential to maintain reliable operation of the control system.  
At the same time, the decisions made to determine the wireless resource $\alpha_t$ and  the control input $u_t$ are intrinsically linked by the control and communication structure outlined in equations \eqref{eq:ind_plant_dyn}-\eqref{eq:plant_switched_prob}. In the following section we formulate the co-design of control and communication policies that optimize system performance in the presence of a wireless fading channel. 

\subsection{Optimal design over fading channels}

Under the centralized shared access model shown in Figure \ref{fig:wl_control_sys}, we consider the optimal design of both the control system and wireless network. 
As seen in \eqref{eq:plant_switched_prob}, the performance of a control policy is closely tied to the state of the communication channel, given by the fading condition $h_t$, as well as the resources allocated by the AP, given by $\alpha_t$. Under limited resources, the optimal design problem consists of both a controller design and resource allocation design that cooperatively keep all the plants in desirable states. For the latter case, resource allocation is clearly guided by need,  in terms of various plant states $x^{(i)}_t$ for $i=1,\hdots,m$, and guided by cost, in terms of fading states $h^{(i)}_t$ for $i=1,\hdots,m$. As both of these change over time, we search for a resource allocation function $\alpha(h,x)$ that, given current fading states  $h_t := [h_t^{(1)}, \hdots, h_t^{(m)}]$ and plant states $x_t := [x_t^{(1)}, \hdots, x_t^{(m)}]$, distributes resources available in $\mathcal{A}$ while respecting the constraints on resource utilization or performance of the plants, which  can be represented by some utility function $l(x, h, \alpha ): \mathbb{R}^{mp} \times \mathbb{R}^m \times \mathbb{R}^m \to \mathbb{R}^r$. Note that  the generic structure of the resource allocation function $\alpha(h,x)$, set $\mathcal{A}$ and utility function  $l(x, h, \alpha)$ may encompass a wide range of problems in communication systems, e.g., power allocation, user scheduling and frequency allocation.

The controller, on the other hand, consists of a mapping $u(h,x): \mathbb{R}^m \times \mathbb{R}^{mp}  \to \mathbb{R}^{mq} $ that, given plant states and channel conditions, returns control signals to be executed by the plants. Observe that, relative to the standard control policy defined in \eqref{eq:plant_switched}, we have expanded  the policy to include the channel state $h$ as input. This is necessary due to the fact that the channel state will have an impact on the resource allocation as governed by the communication policy $\alpha(h,x)$, which in turn effects the probability of closing the control loop. 
 By utilizing the channel state as an input, the control policy is capable of considering these effects in determining its control action.  We also assume that the AP and RC only have access to noisy observations of the plant states and channel conditions, i.e., control and allocation decisions are computed based on
\begin{equation}
[ \tilde{h_t}; \tilde{x_t} ] =  \left[ h_t;x_t  \right] + w_t^{(o)}, 
\label{eq:fading_state_estimate}
\end{equation}
where the observation noise $w_t^{(o)}$ is a zero-mean disturbance with covariance $W^{(o)}$. The observation noise stands for potential measurement errors of channel and control states, and observing measurement errors during training allows the co-design policy to account for imperfectness of channel or control state information during deployment. Note that the observation noise makes the resource allocation problem a partially observable Markov decision process; we thus restrict our attention to Markovian control and allocation policies in order to solve the codesign problem in a model-free manner using a primal-dual reinforcement learning approach. We further assume control actions are restricted to the set $\mathcal{U}$, which may define, e.g., a range of values, i.e. $\mathcal{U}:=[u_{\min}, u_{\max}]^m$. 

The optimal design of a wireless control system can thus be formulated as the joint selection of resource allocation policy $\alpha(h,x)$ and control policy $u(h,x)$ to keep plants operating around an equilibrium point or desirable state under a given set of wireless channel conditions and resource constraints.
 The performance of the resource allocation and control policies is measured by a quadratic cost that penalizes deviations of the plants states from the equilibrium point (assumed to be $0$ without loss of generality) and large control efforts. 
 The optimal co-design problem takes the form
\begin{equation}
 \begin{aligned}
    P^*=\minimize_{\pi(\cdot)} \, &\mathbb{E}^{\pi(\cdot)}_{x_0}  \left[\sum_{t = 0}^{\infty}\gamma^t \left( x_t^\trp Q x_t  + u_t^\trp R u_t \right) \right] \\
  \text{s.t.} \, &\mathbb{E}_{x_0}^{\pi(\cdot)} \left[ \sum_{t = 0}^\infty \gamma^t l^{(i)}(x, h, \alpha) \right] \leq 0 , i = 1, \dots, r \\
  & \pi = \left[ \alpha(h, x) \in \mathcal{A};  u(h, x) \in \mathcal{U} \right] \\
 \end{aligned}
 \label{eq:constrained_optimal_prob}
\end{equation}
with $Q \geq 0$, $R > 0$ cost matrices weighing the contribution of the plants states and control actions, respectively, to the overall cost; $\gamma \in [0, 1]$ a discount factor, and $l^{(i)}(\cdot)$ the $i$-th component of the constraint. At each time $t$, the AP uses resource $\alpha_t^{(i)} = [\alpha(h_t, x_t)]_i$ to send the control signal $u_t^{(i)} = [u(h_t, x_t)]_i$ back to plant $i$. The communication exchange subsequently occurs with success rate given by $v(\varsigma^{(i)}(h^{(i)}_t, \alpha_t^{(i)}))$ and plant $i$ evolves via the switched dynamics in \eqref{eq:plant_switched_prob} accordingly. Having access to information about both channel conditions and plant states allows the AP to balance communication (adjusting channels' transmission conditions and assigned resource, for example) and control needs (prioritizing more unstable plants, for example) while simultaneously selecting a corresponding control action.  We further emphasize in \eqref{eq:constrained_optimal_prob} that the set of permissible resource allocation actions $\mathcal{A}$ and associated SNR function $\varsigma(h,\alpha)$ moreover defines the communication architecture and network deployment scenario under which the control systems operate.
  We also emphasize that the objective of this approach is not to train a single policy that works well across different realizations, but rather to propose a method that can learn policies tailored to a specific wireless control system.

Observe that the co-design problem in \eqref{eq:constrained_optimal_prob} involves optimizing a performance metric over the allocation and control functions $\alpha(h, x)$ and $u(h, x)$ while satisfying the constraints established by $l(x, h, \alpha)$. 
This leads to a constrained, infinite-dimensional optimization problem that is, in practice, 
very challenging to solve --- even in the absence of constraints. Existing approaches have aimed to approximate solutions for linear systems by finding or imposing a separation principle between the communication and control layer decisions \cite{molin2009lqg, ramesh_schedulers_2013, gatsis_optimal_2014}. Those approaches fundamentally rely on model knowledge --- i.e. plant dynamics, communication network models, and channel distributions ---  
and optimal separation of communication and control policies often requires decentralized information structures.
The lack of model knowledge in practice and the challenging nature of the problem have previously motivated the use of learning techniques for resource allocation and wireless control, cf., e.g., \cite{demirel_deepcas:_2018,leong_deep_2018}. Here, on the other hand, we leverage constrained reinforcement learning to design algorithms that can learn feasible resource allocation and control policies and be implemented without any model knowledge of plants or communication models. Before proceeding, we illustrate in the following examples the form taken by the constraints and associated resource allocation set $\mathcal{A}$ in some commonly considered communication and control models.
 
 \begin{example}[Sum-Power Constraints]\label{ex_power_allocation}
 	Traditionally, in wireless communication networks, one wants to minimize the total power used by the AP, in which case the long-term constraint can be written as
 	\begin{equation}
 		\mathbb{E}_{x_0}^{\pi(\cdot)} \left[  \sum_{t = 0}^{\infty} \gamma^t \left( \sum_{i = 1}^m \alpha^{(i)}_t %- \frac{p_{\max}}{T}
 		\right) \right] \leq p_{\max},
 	\end{equation}
 	with $p_{\max}$ the maximum power budget to be used by the AP. 
 \end{example}

 \begin{example}[Region Constraints]\label{ex_fdma} One often wants to limit the amount of time each plant $i$ in the WCS operates outside a region of interest $\mathcal{R}^{(i)} \subset \mathbb{R}^p$ of the state space, that is,
		\begin{equation}
			\mathbb{E}_{x^{(i)}_0}^{\pi(\cdot)} \left[ \sum_{t = 0}^\infty \gamma^t  \mathbbm{1} \{x_t^{(i)} \not \in \mathcal{R}^{(i)} \} \right] \leq T_i
		\end{equation}
	with $\mathbbm{1}$ standing for the indicator function and $T_i$ a bound on the time spent by plant $i$ outside $\mathcal{R}^{(i)}$.
 	On the other hand, given that all transmissions are handled by the single AP, the resource allocation consists of a shared power resource with total budget $\alpha_{\max}$. As such the resource constraint set can be given by the scaled $m$-dimensional simplex
 	\[
 	\mathcal{A} := \left\{\alpha \in \mathbb{R}^m_+ \, : \, \sum_{i = 1}^m \alpha^{(i)} \leq \alpha_{\max} \right\}.
 	\]  
 \end{example}

\begin{remark}
	 In this paper we present an approach to learn feasible policies that respect the constraints during execution, and not necessarily during training. That is different from the traditional view in areas such as safe reinforcement learning, where one typically wants to ensure that the agent does not visit unsafe regions during  training \cite{Cheng2019EndtoEndSR}. 
\end{remark}

%%%%%%%%%%%%%%%%%%%%%%%%%%%%%%%
%%% SECTION : Co-Design via RL  %%%
%%%%%%%%%%%%%%%%%%%%%%%%%%%%%%%
\section{Co-Design via Reinforcement Learning}
\label{sec:resource_RL}
Since the control and resource allocation policies are assumed to be Markovian, one can view the co-design problem in  \eqref{eq:constrained_optimal_prob} as a constrained Markov decision process (MDP), which can be solved in a model-free manner using reinforcement learning. 
Reinforcement learning (RL) represents the intuition that learning occurs in interaction with the environment: an agent executes an action, receives a one-step cost from the environment, transitions to a new state and continues to explore the environment while optimizing the cumulative cost. 
Agent actions $a$ are given by some stochastic policy $\pi(a|s)$ defined on a policy space $\Pi$, and each transition from a state $s_t$ to a new state $s_{t + 1}$ incurs a cost per stage $r_t$.
The problem consists then in finding the policy $\pi^*(\cdot)$ that achieves the minimum cumulative cost of those one-step transitions \cite{sutton_reinforcement_learning,hernandez-lerma_discrete-time_1996}. 

Consider then a centralized co-design agent located at the edge device in Figure \ref{fig:wl_control_sys} that  controls both the AP --- making communication decisions --- and the RC --- making control decisions. The decision-making agent takes as inputs a global state $s_t \in \mathcal{S} \subseteq \mathbb{R}^{m(  1 + p)}$ containing noisy plant states and channel conditions,
\begin{equation}
	s_t = \left[ \tilde{h}_t; \tilde{x}_t \right] = \left[ \tilde{h}_t^{(1)}, \hdots, \tilde{h}_t^{(m)}; \tilde{x}_t^{(1)}, \hdots, \tilde{x}_t^{(m)}\right].
	\label{eq:joint_design_state_RL}
\end{equation}
At each time instant the agent outputs an action $a_t \in \Gamma = (\mathcal{A} \times \mathcal{U}) \subseteq \mathbb{R}^{m (n+q)}$ that consists of a resource allocation decision to be implemented by the AP and control signals sent to the plants, 
\begin{equation}
	a_t = \left[ \alpha_t^{(1)}, \dots, \alpha_t^{(m)}; u_t^{(1)}, \dots, u_t^{(m)}  \right],
	\label{eq:joint_design_action_RL}
\end{equation}
with the agent's policy $\pi$ made up of $\alpha(\cdot)$ and $u(\cdot)$. Considering all functions when optimizing the policy is prohibitively expensive, as such it is common to approximate $\pi$ with a stochastic policy $\pi(\cdot; \theta)$ that is parametrized by a finite dimensional parameter $\theta \in \mathbb{R}^{n}$,
\begin{equation}\label{eq_param}
	\pi(a|s) = \pi (a | s; \theta).
\end{equation}
The parametrized co-design problem can then be written as
\begin{equation}
	\begin{aligned}
		& P^*_\theta =  \minimize_{\theta} J(\theta) \\
		& \text{s.t.} \, \mathbb{E}_{x_0}^{\pi(\cdot; \theta)} \left[ \sum_{t = 0}^\infty \gamma^t l(x, h, \alpha) \right] \leq 0  \\
		& \qquad \pi (a | s; \theta) = \left[ \alpha(\cdot; \theta) \in \mathcal{A}; u(\cdot;  \theta) \in \mathcal{U} \right] \\ 
	\end{aligned}
	\label{eq:constrained_optimal_prob_param_single_agent}
\end{equation}
with
\begin{equation}\label{eq_value1}
	J(\theta) = \mathbb{E}^{\pi(\cdot; \theta)}_{x_0} \left[\sum_{t = 0}^{\infty} \gamma^t ( x_t^\trp Q x_t + u_t^\trp R u_t ) \right].
\end{equation}
There are two critical observations to make in the parametrized problem defined in \eqref{eq:constrained_optimal_prob_param_single_agent}. Firstly, while the parametrization in \eqref{eq_param} gives a significant computational and practical advantage by restricting the dimensionality of the co-design policy, this naturally incurs a loss of optimality with respect to the original, non-parametrized problem \eqref{eq:constrained_optimal_prob}. Secondly, note that standard value-based RL methods operate directly on the value function in \eqref{eq_value1} \cite{sutton_reinforcement_learning}, and thus do not address the constraints $l(x,h,\alpha)$ or the constraint sets $\mathcal{A}$ and $\mathcal{U}$. We tackle both of these issues through the utilization of near-universal approximators and Lagrangian duality, respectively, and analyze their impact on the optimality and subsequent stabilization capabilities of the learned co-design policy in the proceeding sub-section.

\subsection{The Lagrangian Dual Problem}
\label{subsec:primal_dual_rl}

To account for the constraint in \eqref{eq:constrained_optimal_prob_param_single_agent}, we formulate the Lagrangian dual of the constrained optimization problem with the introduction of the dual variable $\lambda \in \mathbb{R}^r$ associated to the constraint $l(\cdot)$, 
\begin{equation}
	\begin{aligned}
		\mathcal{L} (\theta, \lambda) &:= J(\theta)
		+ \lambda^\trp \mathbb{E}_{x_{0}}^{\pi(s; \theta)}  \left[  \sum_{t = 0}^\infty \gamma^t l(x_t, h_t, \alpha_t)  \right] \\
		&= \mathbb{E}_{x0}^{\pi(s; \theta)} \left[ \sum_{t = 0}^\infty \gamma^t \tilde{r}_t \right],  
	\end{aligned}
	\label{eq:lagrangian_param}
\end{equation}
where the penalized one-step cost $\tilde{r}_t$ is defined as
\begin{equation}
	\tilde{r}_t :=  x_t^\trp Q x_t + u_t^\trp R u_t + \lambda^\trp  l(x_t,h_t,\alpha(s_t; \theta)),
	\label{eq:dual_cost_penalized}
\end{equation}
which follows from the linearity of expectation. 
Given the Lagrangian in \eqref{eq:lagrangian_param}, the dual function associated to the constrained optimization problem \eqref{eq:constrained_optimal_prob_param_single_agent} is defined as
\begin{equation}\label{eq_dual}
	d_\theta(\lambda) \coloneqq \min_{\theta \in \mathbb{R}^n} \mathcal{L}(\theta, \lambda).
\end{equation}

Observe that, for a fixed $\lambda$, evaluating the dual function in \eqref{eq_dual} corresponds to a standard MDP problem with the objective given by the Lagrangian $\mathcal{L}(\theta,\lambda)$ and one-step cost given in \eqref{eq:dual_cost_penalized}. The inner optimization problem can thus be solved using RL methods --- see \cite{sutton_reinforcement_learning} for a survey on standard RL methods. The complete dual optimization problem consists then in maximizing the dual function with respect to $\lambda$,
\begin{equation}
	D_\theta^* = \max_{\lambda \geq 0 } d_\theta(\lambda) = \max_{\lambda \geq 0 } \min_{\theta} \mathcal{L} (\theta,\lambda).
	\label{eq:dual_opt_problem}
\end{equation} 
Although the original, constrained problem \eqref{eq:constrained_optimal_prob_param_single_agent} is non-convex, the dual function in \eqref{eq_dual} is the minimum of linear functions in $\lambda$ and is thus concave. That makes the dual problem in \eqref{eq:dual_opt_problem} convex, since it consists in maximizing a concave function subject to the convex constraint $\lambda \geq 0$ \cite{boyd-cvx-opt}.
For a given dual parameter $\lambda$, the policy parameter $\theta$ can be recovered as $\theta(\lambda) = \arg \min_{\theta} \mathcal{L}(\theta, \lambda)$.

This dual formulation of the constrained optimization problem provides a framework to handle the resource utilization constraints and learn feasible policies via a bilevel optimization consisting of an inner RL problem of the primal variable $\theta$ in \eqref{eq_dual} and an outer convex optimization problem of the dual variable $\lambda$ in \eqref{eq:dual_opt_problem}. 
Solving the dual, parametrized problem in equations \eqref{eq:lagrangian_param} --- \eqref{eq:dual_opt_problem}, however, does not necessarily imply optimality of the original constrained co-design problem \eqref{eq:constrained_optimal_prob}.
To establish a relation between the value of the unparametrized primal problem and the value of the parametrized dual problem --- and ultimately establish conditions under which the introduction of policy parametrizations and the use of dual methods yield feasible policies capable of stabilizing the WCS --- we first utilize results from the theory of constrained reinforcement learning.
In the following, we assume that the co-design policy is parametrized by a \emph{near-universal} parametrization $\pi$. Such parametrizations are defined as those that can approximate any policy or function in $\Pi$ within a certain bound. We define this class formally as follows: 
\begin{definition}[\cite{eisen_learning_2018}] Let $\mathcal{P}(\mathcal{S})$ the space of probability measures on $( \Gamma, \mathcal{B}( \Gamma))$ parametrized by elements of $\mathcal{S}$, where  $\mathcal{B}( \Gamma)$ are the Borel sets of $ \Gamma$.
	Then, a parametrization $\pi_\theta$ is an $\epsilon$-universal parametrization of functions in $\mathcal{P}(\mathcal{S})$ if, for some $\epsilon > 0$, there exists for any $\pi \in \mathcal{P}(\mathcal{S})$ a parameter $\theta \in \mathbb{R}^r$ such that
	\begin{equation}
		\max_{s \in \mathcal{S}} \int_{ \Gamma} | \pi(a|s) - \pi_\theta(a|s) | da \leq \epsilon.
	\end{equation}
	\label{def:near_universal_param}
\end{definition}
This property is indeed satisfied by a number of commonly used parametrizations in reinforcement learning, such as deep neural networks and reproducing kernel Hilbert spaces \cite{paternain2019constrained}. The sub-optimality of dual solutions obtained in \eqref{eq:dual_opt_problem} with  $\epsilon$-universal parametrizations can be established in the following theorem, adapted from \cite{paternain2019constrained}. 

\begin{theorem} Assume that Slater's condition holds for \eqref{eq:constrained_optimal_prob}, i.e, there exists a policy $\pi^*(h,x)$ that renders the constraint in \eqref{eq:constrained_optimal_prob} strictly feasible.  Further assume that the one-step cost $r_t$ is bounded  by $B_r > 0$, and that the constraints $l^{(i)}$ are bounded by $B_{l_i} > 0$. Define $B_l = \max_{i = 1, \dots, r}B_{l_i}$ and let $\lambda_\epsilon^*$ be the solution to the dual problem associated to the perturbed primal problem \eqref{eq:primal_perturbed} for perturbation $\xi_i = B_r \epsilon / (1 - \gamma)$. Then, 
	\begin{equation}
		P^* \leq D_\theta^* \leq P^* + \left( B_{r} + \| \lambda_\epsilon^* \|_1 B_l \right) \frac{\epsilon}{1 - \gamma}, 
		\label{eq:dual_param_bound}
	\end{equation}
	where $P^*$ is the optimal value of \eqref{eq:constrained_optimal_prob}, and $D_\theta^*$ the value of the parametrized dual problem \eqref{eq:dual_opt_problem}. 
	\label{theorem:dual_param_bound}
\end{theorem}
\begin{proof}
	See \ref{proof:dual_param_bound}.
\end{proof}

Theorem \ref{theorem:dual_param_bound} linearly relates the approximating capabilities of $\epsilon$-universal parametrizations  to the sub-optimality of the parametrized dual optimization. This characterization of the sub-optimality is important as parametrizing the resource allocation function makes the dimensionality of the learning problem finite, while working in the dual domain provides a framework for handling the constraints of the original co-design problem with standard RL methods. Thus, we may parametrize and work in the dual domain so long as we utilize a sufficiently dense parametrization. Note that the quadratic objective in \eqref{eq:constrained_optimal_prob} will be bounded if we restrict the state and action spaces. Let, for example, $x^{(i)} \in [-x_{\max}, x_{\max} ]^{mp}$, $\bar{x} = [x_{\max}]^{mp}$, $u^{(i)} \in [-u_{\max}, u_{\max}]^{mq}$, $\bar{u} = [u_{\max}]^{mp}$. Then, the bound holds with $B_r = ( \bar{x}^\trp Q \bar{x} + \bar{u}^\trp R \bar{u})$.  

As previously stated, the dual function is convex in terms of the dual parameter $\lambda$ and can hence be solved via the dual-descent procedure, as described in Algorithm \ref{alg:dual_descent}. Dual-descent iteratively updates the co-design policy parameter $\theta$ through the solving of the primal \emph{unconstrained} MDP formed by the Lagrangian in \eqref{eq_dual} using traditional RL methods, followed by a single gradient ascent step of the dual parameter. Conditions under which dual descent converges  can then be established in the following theorem.

\begin{algorithm}[t]
	\caption{Dual-descent}
	\label{alg:dual_descent}
	\DontPrintSemicolon
	\BlankLine
	Initialize $\theta_0, \lambda_0$ \;
	\For{$i$ = $0, 1, \dots$}{
		Solve primal problem: 
		$\theta_{i + 1} \leftarrow \arg \min \mathcal{L}(\theta_i, \lambda_i) $ \;
		
		Update dual variable: $\lambda_{i + 1} \leftarrow [\lambda + \beta_\lambda \nabla_\lambda \mathcal{L}(\theta_i, \lambda_i)]_+ $ \;
	}
\end{algorithm}

\begin{theorem}
	\label{theo:dual_descent_conv}
	Let $\mathcal{L}(\theta, \lambda)$, with $\lambda \in \mathbb{R}^r$, be the Lagrangian associated to the parametrized problem \eqref{eq:constrained_optimal_prob_param_single_agent} and assume that the minimization of $\mathcal{L}(\theta, \lambda)$ in Algorithm \ref{alg:dual_descent} can be solved exactly. Then, under the Assumptions of Theorem \ref{theorem:dual_param_bound}, dual descent with step size $\beta_\lambda$ converges to a neighborhood of $P^*$ for any $\epsilon > 0$.
\end{theorem}
\begin{proof}
	Convergence of dual descent is a well-established result in the literature, see, e.g., \cite[Theorem 5]{Ribeiro2012}.
\end{proof}

Theorem \ref{theo:dual_descent_conv} establishes that the dual descent in Algorithm \ref{alg:dual_descent} is guaranteed to converge to a point close to the solution of the \emph{unparameterized} co-design problem in \eqref{eq:constrained_optimal_prob}. This follows from the convexity of the dual problem in \eqref{eq_dual} and the closeness of parameterized and unparameterized solutions established in Theorem \ref{theorem:dual_param_bound}.  While the implementation of dual descent indeed requires the solving of an unconstrained MDP as in Step 3, this result establishes that the inclusion of long-term constraints adds no \emph{additional} optimization complexity relative to traditional unconstrained RL.
Thus, we may utilize RL as a framework for handling constrained resource allocation and control design in WCSs via the addition of a convex gradient step in the dual parameter.
The sub-optimality incurred through the parametrized dual formulation, however, poses the question of whether a near-optimal policy parametrized by an $\epsilon$-universal parametrization is still capable of stabilizing the wireless control system, if so does the optimal, non-parametrized solution to the constrained co-design problem. It turns out that the optimality loss will affect the probability with which the trajectories of the wireless control system will stay bounded in a region around the equilibrium point, as established in Theorem \ref{theo:approx_sol_stab}. In the following, $\lambda_{\min}(Q)$ stands for the minimum eigenvalue of $Q$, and the value function associated to the optimal codesign problem \eqref{eq:constrained_optimal_prob} is given by
\begin{equation}
	\begin{aligned}
		V(x) &\coloneqq \mathbb{E}^{\pi^*(\cdot)}_{x}  \left[\sum_{t = 0}^{\infty}\gamma^t \left( x_t^\trp Q x_t  + u_t^\trp R u_t \right) \right]. \\
	\end{aligned}
\end{equation}
Similarly, $V_{\theta^*}$ is the value function under a policy $\pi^*_\theta$ parametrized by an $\epsilon$-universal parametrization.

\begin{theorem}
	\label{theo:approx_sol_stab}  Assume that $x^{(i)} = 0, i = 1, \dots, m$ is an equilibrium point for the wireless control system described in equations  \eqref{eq:ind_plant_dyn}-\eqref{eq:plant_switched_prob}, that is, $x^{(i)}_t = 0$ for all $t \geq t_0$ if $x_{t_0}^{(i)}=0$ and $u_t^{(i)} = 0$, and that $V(\cdot)$, $V_{\theta^*}(\cdot)$ are continuous with $V(0)= 0$, $V_{\theta^*}(0) = 0$. Assume also that dual descent converges to a local minimum satisfying KKT conditions and that the optimal non-parametrized solution of \eqref{eq:constrained_optimal_prob} stabilizes the wireless control system. In particular, assume that the value function associated to the optimal, non-parametrized policy satisfies
	\begin{equation}
		\mathbb{E} \left[ V(x_{t + 1})\right ] - V(x_t)  \leq 0 
		\label{eq:stab_vf_lyap_cond}
	\end{equation}
	for any $x_t \in D_r \coloneqq \{ x \in \mathbb{R}^{mp} : \|x\| < r \} \subset \mathbb{R}^{mp}$, $r > 0$. Then, under  an $\epsilon$-universal parametrization, 
	there exists $\delta(\zeta, \nu ,t_0) > 0$ for every $\zeta > 0$ and $\nu > 0$ such that
	\begin{equation}
		P \left\{ \|x_t\| < \nu \right\} \geq 1 - \left(\zeta + \frac{ B_{r} + \| \lambda_\epsilon^* \|_1 B_l  }{(1 - \gamma) \lambda_{\min}(Q)\nu^2} \epsilon \right), t \geq t_0,
		\label{eq:stab_prob}
	\end{equation} 
	when $\|x_0\| < \delta$. 
	
\end{theorem}
\begin{proof}
	See \ref{proof:approx_sol_stab}.
\end{proof}

Theorem \ref{theo:approx_sol_stab} shows that if the optimal non-parametrized solution of the co-design problem \eqref{eq:constrained_optimal_prob} stabilizes the WCS, so will the solution of the dual parametrized problem in \eqref{eq:dual_opt_problem}, with the optimality loss of  an $\epsilon$-universal parametrization affecting the probability with which the trajectories of the WCS are contained to a region around the equilibrium point. Thus, the ``density'' of the parametrization, as captured by $\epsilon$, can be directly related to the degree to which the co-design policy stabilizes the system. 

\begin{remark}
	As shown in Lemma 1 in \cite{eisen_learning_2018}, the entire class of deep neural networks is an $\epsilon$-universal parametrization for all $\epsilon > 0$.  Then, if in addition to the conditions outlined in Theorem \ref{theo:approx_sol_stab}, we further assume the dual variable associated to the perturbed problem, $\| \lambda_\epsilon^* \|_1$, to be finite,  the bound between the value of the primal problem, and the value of the parametrized dual problem in Theorem \ref{theorem:dual_param_bound} would hold for all $\epsilon > 0$, since $B_r$ and $B_l$ are both finite. Under the assumptions of the theorem, making $\epsilon$ arbitrarily small --- as in the case of an arbitrarily large neural network \cite{hornik_1989} --- would imply that equation \eqref{eq:stab_prob} holds with 
	\begin{equation}
		P \left\{ \|x_t\| < \nu \right\} \geq 1 - \zeta, t \geq t_0,
		\label{eq:cor1_stab_prob}
	\end{equation}
	meaning that the parametrized solution in this case renders the system stochastically stable in the sense of \cite[Definition 1]{li_stability_2013}.
\end{remark}

We point out that the dual-descent method detailed in Algorithm \ref{alg:dual_descent} can however be a computationally prohibitive procedure, as it requires an inner-loop RL optimization to be performed for each iterative value of $\lambda$.  We proceed in the following section to derive a more computationally practical procedure that combines traditional, model-free RL methods with dual descent to find solutions to the parametrized dual problem. 

%%%%%%%%%%%%%%%%%%%%%%%%%%%%%%%
%%% SECTION : Model-Free Learning  %%%
%%%%%%%%%%%%%%%%%%%%%%%%%%%%%%%
\section{Primal-Dual Reinforcement Learning}\label{sec_learning}
We present a practical iterative method termed \emph{primal-dual reinforcement learning (RL)}, which alternates RL updates with a dual parameter update in a single loop procedure. As we detail in this section, the primal-dual RL approach combines the model-free capabilities of RL with the constraint satisfaction capabilities of dual methods. 

Model-free reinforcement learning methods can be classified into value-based and policy-based methods \cite{sutton_reinforcement_learning}. Value-based methods approximate the value or state-value function and then select actions that minimize the estimated long-term cost.
Policy-based methods, on the other hand, directly learn a policy $\pi(\cdot; \theta)$ by updating the policy parameter $\theta$ in order to minimize sampled costs. Actor-critic algorithms, in turn, combine features of policy-based and value-based reinforcement learning to reduce variance and improve convergence: the agent learns both the policy, called the \emph{actor}, and the estimated long-term reward, called the \emph{critic}. In the following we utilize actor-critic methods in the proposed primal-dual algorithm.

 For $\lambda$ fixed, the Lagrangian $\mathcal{L}(\theta, \lambda)$ reduces to an unconstrained MDP, and the parameters of the resource allocation or control policy can be updated via approximate gradient descent 
\begin{equation}
\theta_{j + 1} = \theta_j - \beta  \hat{\nabla}_\theta \mathcal{L}(\theta, \lambda),
\label{eq:approximate_gradient_descent}
\end{equation}
where  $\beta$ is the learning rate and $\hat{\nabla}_\theta \mathcal{L}(\theta, \lambda)$ an estimate of the gradient of $\mathcal{L}(\theta, \lambda)$ with respect to $\theta$ \cite[ch. 13]{sutton_reinforcement_learning}. 
The critic, on the other hand, approximates the value function with a function $\rho(\cdot; \eta)$ parametrized by $\eta$. The parameter $\eta$ is optimized to approximate the Lagrangian function in terms of mean squared error, i.e.,
\begin{equation}
\eta^* := \arg \min_{\eta} \mathbb{E}^{\pi(\cdot; \theta)}_{x_0} (\mathcal{L} (\theta, \lambda^*) - \rho(s; \eta))^2.
\label{eq:value_parameterization}
\end{equation}

Here, in particular, we adapt the Advantage Actor-Critic (A2C) algorithm \cite{mnih_a3c_2016}, which relies on synchronous execution of multiple agents in parallel realizations of the environment using model free estimates of the gradients in \eqref{eq:approximate_gradient_descent} and \eqref{eq:value_parameterization}, to the primal-dual setting. 
   At each iteration, the policy is updated according to
 \begin{equation}
 \theta_{j + 1} = \theta_j - \beta_\theta \sum_{t = 0}^T \nabla \log \pi (a_t | s_t; \theta_t) A(s_t, a_t; \theta, \eta),
 \label{eq:policy_update}
 \end{equation}
 with $A(s_t, a_t; \theta, \eta)$ an estimate of the advantage function, 
defined as the difference between $Q^{\pi(\cdot; \theta)}(s,a)$, the 
 expected return starting from state $s$ given action $a$, and $J^{\pi(\cdot; \theta)}(s)$, the expected return from state $s$ given the current policy,
\begin{equation}
A^{\pi(\cdot; \theta)}(s,a) = Q^{\pi(\cdot; \theta)}(s,a) - J^{\pi(\cdot; \theta)}(s).
\label{eq:advantage}
\end{equation}
Given the critic learning rate $\beta_\eta$, the value function approximator is updated by
\begin{equation}
\eta_{j} = \eta_j - \beta_\eta \sum_{t = 0}^T \frac{\partial}{\partial \eta}  \left(\rho (s; \eta)  - \tilde{R}_t \right)^2
\label{eq:value_update}
\end{equation}
so as to minimize the squared difference between the estimated value function and the sampled cost-to-go, 
\begin{equation}
	\tilde{R}_t  = 
	\sum_{\tau = t}^{ \bar{t} - 1}  \gamma^{\tau - t} \tilde{r}_\tau + \gamma^{t_{\max}} R_{\bar{t}}, t = \bar{t} - 1, \dots, \bar{t} - t_{\max},  
	\label{eq:cost_to_go}\end{equation}
with $\bar{t} = k t_{\max}$, $k \geq 1$, and $t_{\max}$ the interval between updates. $ R_{\bar{t}} = 0$, if $\bar{t} = T$, and $ R_{\bar{t} } = \rho (s_{\bar{t} }; \eta)$ otherwise \cite{mnih_a3c_2016}.
The dual variable $\lambda$ can be updated via approximate gradient ascent, 
\begin{equation}
\begin{aligned}
\lambda_{i + 1} &= [ \lambda_i + \beta_{\lambda} \hat{\nabla}_{\lambda} \mathcal{L}(\theta, \lambda) ]_{+}, \\
\hat{\nabla}_{\lambda} \mathcal{L}(\theta, \lambda) &= \mathbb{E}^{\pi(\cdot; \theta)} \left[ \sum_{t = 0}^T l(x, h, \alpha) \right].
\end{aligned}
\label{eq:lambda_update}
\end{equation}

\begin{algorithm}[t]
     \caption{Primal-dual RL 
      }
     \label{alg:deep_pg}
  \DontPrintSemicolon
   \KwResult{Resource allocation and control policies.}
      \BlankLine
   \For{$i$ = $1, \dots, $}{
   Generate episodes; \;
     \For{ t = 1, \dots, T}{
Collect samples from $N$ threads; \;
    \If{$t = k t_{\max}, k = 1, 2, \dots$}{ 
   	Compute cost-to-go \eqref{eq:cost_to_go} \;
   	Compute advantage estimate \eqref{eq:advantage};\;
	Update policy \eqref{eq:policy_update};\;
   	Update value function estimate \eqref{eq:value_update}; \;
   		}
	}
	Update dual variable \eqref{eq:lambda_update}; \;
   }
  \end{algorithm}

 The resulting procedure for designing model-free resource allocation and control policies in WCSs is presented in Algorithm \ref{alg:deep_pg}.
The implementation considers $N$ simultaneous realizations of the WCS per episode, with each realization having a simulation horizon $T$.  Each realization starts from initial states $s_ 0 = [x_0; h_0]$ with $x_0$ sampled from a standard normal distribution and $h_0$ sampled from $\chi(h)$, a probability distribution representing the distribution of the wireless fading conditions.
 At each time step, allocation and control decisions for all users are sampled from the corresponding policy $\pi(\cdot; \theta)$. 
  The plants' states evolve according to \eqref{eq:ind_plant_dyn} and fading conditions are sampled again from $\chi(h)$. After $t_{\max}$ time steps, we compute the cost-to-go $\tilde{R}_t$ and estimate the advantage function $A_t$. The advantage function  and the cost-to-go from $N$ simultaneous realizations are aggregated to estimate the gradients of the policy and value networks, and the corresponding parameters are updated accordingly. 
At the end of each episode, an estimate of the constraint violation, $\hat{\nabla}_\lambda \mathcal{L}(\theta, \lambda)$, is used to update the dual variable. Noisy observations are used during training to account for possible measurement errors.  The proposed primal-dual approach can easily be extended to other RL algorithms that rely primarily on policy parametrizations by changing the primal update in Algorithm \ref{alg:deep_pg} accordingly. We further point out that the training procedure in Algorithm \ref{alg:deep_pg} is performed offline, and as such the convergence time is not critical.  If the underlying channel distribution model changes over long periods of time, fine tuning the policy may be necessary in practice, but that can be done online. The inference time of the resource allocation and control decisions, on the other hand, is minor as it consists in simply using the previously trained policy to compute control and allocation decisions based on the current fading and control states. We assume the inference time to be shorter than the coherence time of the channel. In practical scenarios, the co-design policy can first be trained using simulations based on data collected from practical systems, and then deployed in a physical system so that the co-design policy can be fine tuned. At each time step, the codesign agent should collect control and fading states, sample an action from its current policy, execute that action, and then observe the corresponding cost --- namely, the control effort and the weighted norm of control states in \eqref{eq:constrained_optimal_prob}. The constraint violation is measured by comparing the resource used to send the control signals to the plants against the available power budget. After a pre-determined number of time steps, the approximate policy, critic and dual gradients can be computed from the collected information, and its parameters can be updated accordingly.

Utilizing the common framework of \emph{deep} RL, we parametrize the resource allocation and control policies with deep neural networks --- known universal approximators \cite{hornik_1989}.  Both the actor network $\pi(\cdot; \theta)$ and critic network $\rho(\cdot; \eta)$ take as inputs the state $s_t$ defined in \eqref{eq:joint_design_state_RL} and output parameters that characterize $2m$ independent Gaussian policies with means $\mu_\alpha^{(1)}, \dots,\mu_\alpha^{(m)}$, $\mu_u^{(1)}, \dots,\mu_u^{(m)}$ 
  from which resource allocation decisions $\alpha^{(1)},\hdots,  \alpha^{(m)}$ and control actions $u^{(1)}, \hdots, u^{(m)}$ are drawn. These means are collected in the output vector $y(x,h)$ which is computed with a standard neural network,
  \begin{equation}
y^{(k)}(x,h) = \phi_L \left(C_L  \phi_{L-1} \left( \dots \phi_1 \left( C_1 s + b_1 \right) \right) + b_L  \right) 
\label{eq:nn_output}
\end{equation}
 with the matrix $C_l$ and vector $b_l$ representing the weights of the linear combination at each hidden layer $l = 1, \dots, L$, and $\phi(\cdot)$ a pointwise nonlinearity. 
The learnable parameters here correspond to the weights of the linear combinations 
 and standard deviations $\sigma_\alpha^{(1)}, \dots,\sigma_\alpha^{(m)}$, $\sigma_u^{(1)}, \dots,\sigma_u^{(m)}$ for the Gaussian policy distribution.

 \begin{remark}\label{remark_constraints}
The generic structure of neural networks further allows us to design parametrizations that adhere to the resource allocation structure $\mathcal{A}$ and control constraints $\mathcal{U}$ in \eqref{eq:constrained_optimal_prob}.  In particular, the output layer activation $\phi_L$ can in most cases be constructed to embed such structure on the outputs. For example, resource allocation decisions $\alpha$ satisfying $\sum_{i} \alpha^{(i)} \leq \alpha_{\text{max}}$ can be addressed with a softmax-type normalization of the outputs. Likewise, interval constraints such as $\mathcal{U} = [u_{\min}, u_{\max}]^m$ can be addressed via a sigmoid output layer. 
 \end{remark}

%%%%%%%%%%%%%%%%%%%%%%%%%%%%%%%%%%%%%%%%
%%% SECTION : Numerical Experiments  %%%
%%%%%%%%%%%%%%%%%%%%%%%%%%%%%%%%%%%%%%%%

%%%%%%%%%%%%%%%%%%%%%%%%%%%
%%%%%%%%%% FIGURE %%%%%%%%%
%%%%%%%%%%%%%%%%%%%%%%%%%%%
\begin{figure*}[htb]
    \centering
    \begin{minipage}{0.33\textwidth}
        \centering
       \includegraphics[width=\linewidth]{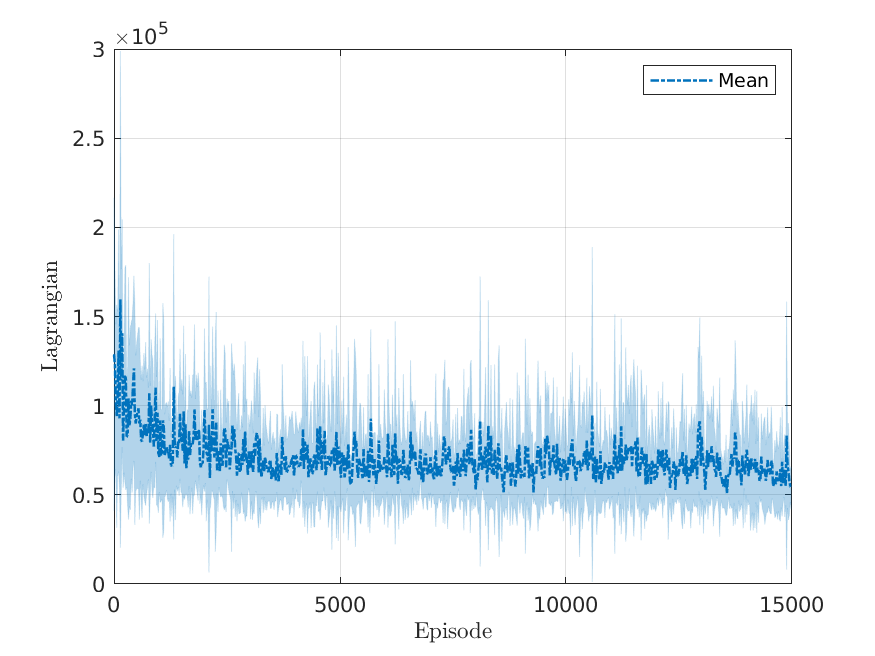}
    \end{minipage} 
        \hfill %
        \begin{minipage}{0.33\textwidth}
        \centering
      \includegraphics[width=\linewidth]{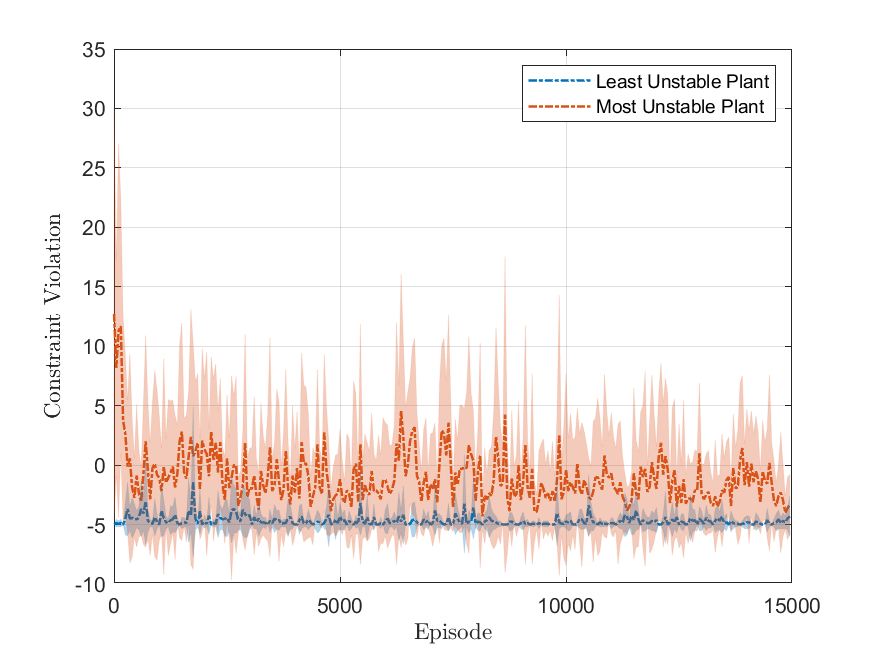}
    \end{minipage}% 
    \hfill %
        \begin{minipage}{0.33\textwidth}
        \centering
      \includegraphics[width=\linewidth]{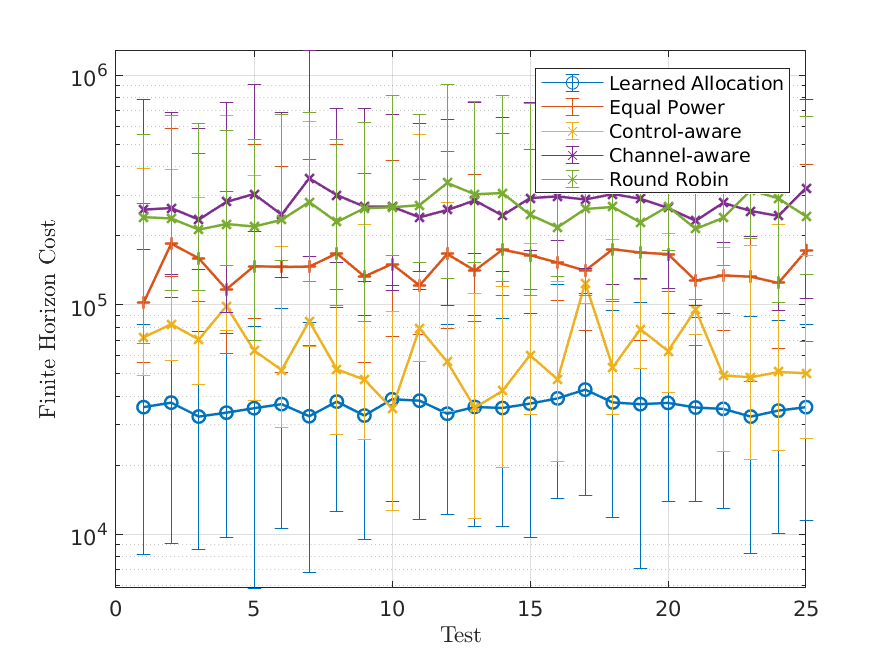}
    \end{minipage}% 
    \caption{Linear system: (left) Lagrangian  and  (center) constraint violation per training episode. (Right) Cost comparison between the learned allocation policy (blue), prioritizing unstable plants (yellow), dividing power equally (orange), round robin (green) and prioritizing channel conditions (purple).
     }
    \label{fig:linear_system_safety_sims}
\end{figure*}
%%%%%%%%%%%%%%%%%%%%%%%%%%%
%%%%%%%%%%%%%%%%%%%%%%%%%%%
%%%%%%%%%%%%%%%%%%%%%%%%%%%

\section{Numerical Experiments}
\label{sec:num_exp}

We now present numerical experiments to illustrate the use of the proposed model-free approach to design control-aware resource allocation and channel-aware control policies in WCSs. 
We created custom environments on OpenAI Gym \cite{brockman2016openai} and customized a standard implementation of A2C from Stable Baselines  \cite{stable-baselines} to account for the dual learning step in Algorithm \ref{alg:deep_pg}. Unless indicated otherwise, policies were parametrized with a standard multilayer neural network made up by two hidden layers with 64 hidden units each and trained over $N = 16$ simultaneous realizations of the WCS, with a discount factor $\gamma = 0.99$. 
 Parameters were tuned heuristically.
The AP and the RC have access to noisy estimates of the control states and channel conditions, cf. equation \eqref{eq:fading_state_estimate}.
Initial states of the control plants are sampled from a standard normal distribution. Plants are randomly dispersed around the AP collocated with the RC. The distance $d_i$ between plant $i$ and the RC determines the slow fading component, $h_s = d_i^{-p_l}$, with $p_l$ the path loss. The fast fading component $h_f^{(i)}$, on the other hand, follows a Rayleigh distribution $\chi(h)$ with parameter $\lambda_h$, leading to
$
h^{(i)} = h_s^{(i)} \times h_f^{(i)}, h_f^{(i)} \sim \chi(h).
$  The probability of successfully receiving an information packet is given by 
\begin{equation}
v(\varsigma(\alpha, h)) = 1 - \exp (- \varsigma(\alpha, h)).
\end{equation}

\paragraph{Power Allocation: Instantaneous Power and Region Constraints}
Here the agent must learn a resource allocation policy to distribute transmission power among $m = 10$ plants sharing a wireless communication network. In this first scenario we assume that the control action is computed by a standard linear quadratic regulator (LQR controller). The allocation policy must satisfy an instantaneous resource constraint,
\[ 
\sum_{i = 1}^{m} \alpha^{(i)}_t(x_t, h_t) \leq  p_{\max}
\] 
on the power utilized to send the control signals back to the plants at each time instant, with $p_{\max} = m$. The objective here is to minimize a finite horizon quadratic cost \eqref{eq:constrained_optimal_prob} with $Q = \mathbb{I}$ and $R = \mathbb{I}$.
 The internal state of each plant \eqref{eq:plant_switched_prob} evolves according to 
\begin{equation}
x_{t + 1}^{(i)} = A^{(i)}x^{(i)}_t + B^{(i)} u_t^{(i)} + w_t,
\label{eq:linear_plants}
\end{equation}
where
\begin{equation}
A^{(i)} = \begin{bmatrix}
-a^{(i)} & 0.2 & 0.2 \\ 0 & -a^{(i)} & 0.2 \\ 0 & 0 & -a^{(i)} 
\end{bmatrix}; \quad
B^{(i)} = \mathbb{I},
\label{eq:matrices_value_power_allocation}
\end{equation}
with $\mathbb{I}$ standing for the identity matrix and $a^{(i)}$ sampled from a uniform distribution, $a^{(i)} \sim \mathcal{U} [1.05, 1.15]$. 
The plants are open-loop unstable but controllable, that is, under ideal communication conditions, the control input can drive the system back to equilibrium. Furthermore, all plants have different dynamics in this scenario, stressing the heterogeneity of the WCS and the different communication needs of each plant. Here, $W^{(o)} = \mathbb{I}$, $p_l = 2$ and $\lambda_h = 1$. 
Besides distributing power among the control signals to be sent to the plants, the agent must also keep each plant in a region $\mathcal{R}^{(i)} = [-15, 15]^p$ with $T_i = 5$, as described in Example \ref{ex_fdma}.
 The plants are randomly placed in a $[-m/4, m/4]^2$ area around the AP. 

The training phase considers episodes of horizon $T = 100$ and learning rates $\beta_\theta = \beta_\eta = 5 \times 10^{-4}, \beta_\lambda = 1 \times 10^{-4}$. To facilitate training, we first train the resource allocation policy to imitate a heuristic that gives more power to plants away from the equilibrium point.  Figure \ref{fig:linear_system_safety_sims} presents the mean and one standard deviation of (left) the Lagrangian and (center) the constraint violation per episode for the least (blue) and most (orange) unstable plants, where we can see that the learned policy converges to a feasible solution after 2000 episodes. 
After the training phase, we compare the performance of the learned allocation policy against baseline resource allocation solutions,  either sending control signals to all the plants at each time instant or prioritizing some plants to transmit at each time. In the latter case, we consider the following heuristics:
\begin{enumerate}
\item channel-aware selection: choose $ m / 3 $ plants with best channel conditions to transmit; %
\item control-aware selection: choose $ m / 3 $ plants furthest away from the equilibrium point to transmit; 
\item round-robin: schedule transmission of  $ m / 3 $ plants per time slot.
\end{enumerate}
To make comparisons fair, the power budget available at each time instant, $p_{\max}$, is equally divided among the transmitting signals.
To compare the learned policy against the baseline resource allocation solutions described above and see how well the learned policy generalizes to longer horizons not seen during training, we consider a horizon $T = 120$ during runtime.
Figure \ref{fig:linear_system_safety_sims} (right) shows a comparison between the learned policy and the baseline solutions mentioned above. The figure shows the mean and range of the finite horizon cost per test, with each test consisting of ten simultaneous realizations starting from different initial states. 
The learned allocation policy, in blue, clearly outperforms the heuristic solutions. 

%%%%%%%%%%%%%%%%%%%%%%%%%%%
%%%%%%%%%% FIGURE %%%%%%%%%
%%%%%%%%%%%%%%%%%%%%%%%%%%%
\begin{figure*}[htb]
    \centering
    \begin{minipage}{0.33\textwidth}
        \centering
       \includegraphics[width=\linewidth]{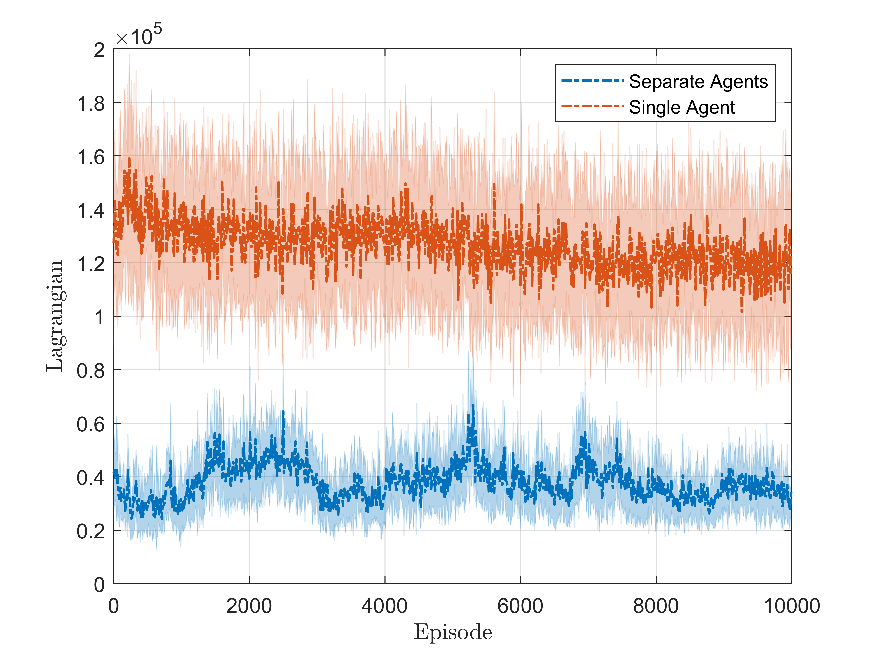}
    \end{minipage} 
        \hfill %
        \begin{minipage}{0.33\textwidth}
        \centering
      \includegraphics[width=\linewidth]{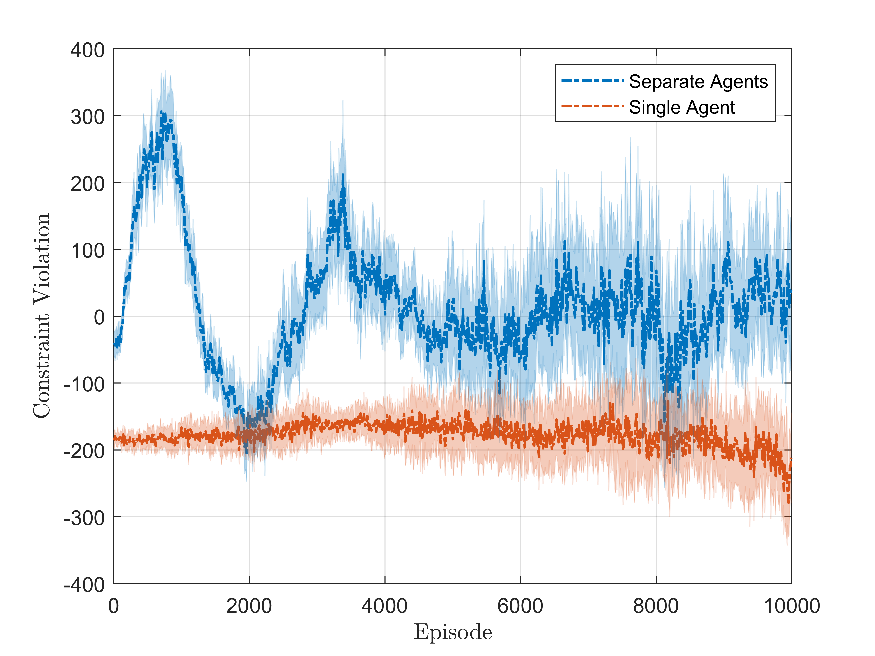}
    \end{minipage}% 
    \hfill %
        \begin{minipage}{0.33\textwidth}
        \centering
      \includegraphics[width=\linewidth]{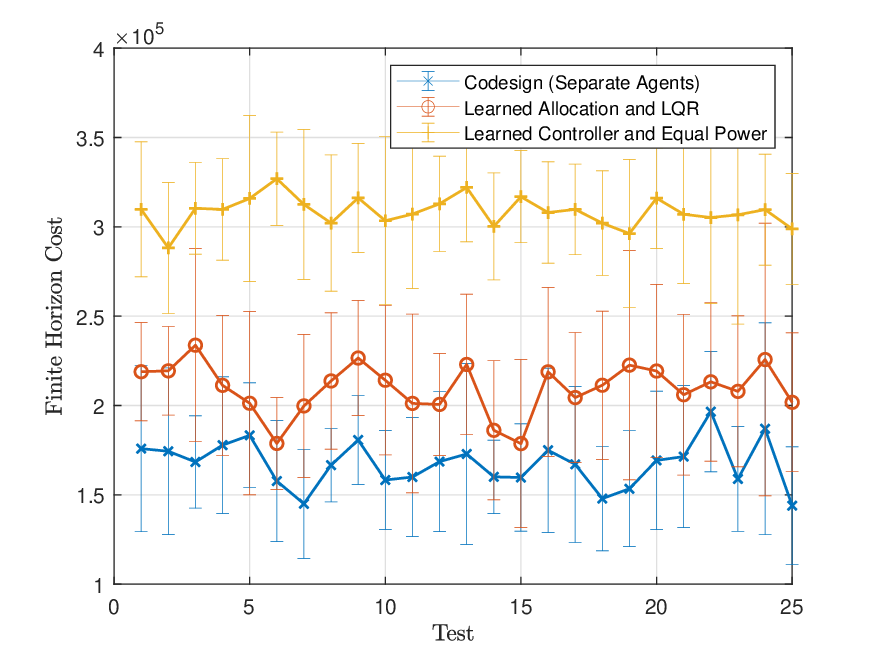}
    \end{minipage}% 
    \caption{Linear system: (left) Lagrangian  and  (center) constraint violation per training episode. (Right) Cost comparison between the model-free codesign policy (blue), learning a resource allocation policy while using a LQR controller based on the plant's model (orange), and learning a control policy while distributing resource equally among the plants (yellow). }
    \label{fig:linear_system_codesign_sims}
\end{figure*}
%%%%%%%%%%%%%%%%%%%%%%%%%%%
%%%%%%%%%%%%%%%%%%%%%%%%%%%
%%%%%%%%%%%%%%%%%%%%%%%%%%%

\paragraph{Co-Design --- Linear System}
Next we evaluate scenarios where one must learn both the resource allocation and control policies. 
An immediate approach  is to consider a single, centralized agent simultaneously learning both policies. That leads to a problem with high dimensionality, however, and learning good policies in this setting is challenging, especially as the number of users in the system grows --- since the number of parameters to be learned depends on the number of plants in the wireless system and on the dimension of the plants states, cf. equation \eqref{eq:nn_output} --- or in the presence of more unstable plants.
To address these practical challenges, we further consider an alternative RL formulation that features a separation in design between the communication and control policies, and compare two approaches to solve the problem, namely 
\begin{enumerate}
\item Single agent: simultaneously learning control and allocation policies (cf. equations \eqref{eq:joint_design_state_RL} - \eqref{eq:joint_design_action_RL});
\item Separate agents: allocation decisions are computed  by an AP agent, based on plants states and channel conditions; control actions are computed by an RC agent, based on plants states, channel conditions and allocation actions.
\end{enumerate}
In the second approach, the states of the agents representing the AP and the RC are given by
\begin{equation}
		s^{\text{AP}}_t = \left[ \tilde{h}_t; \tilde{x}_t  \right];  s^{\text{RC}}_t = \left[ \tilde{h}_t;\tilde{x}_t; \alpha_t  \right]
	\label{eq:codesign_separate_1}
\end{equation}
with $s_t^{AP} \in \mathbb{R}^{m \times (1 + p)}$,  $s_t^{RC} \in \mathbb{R}^{m \times (1 + p + 1)}$. The AP agent outputs allocation decisions taken according to the resource allocation policy $\pi_{\alpha}(\tilde{h},\tilde{x}; \theta_\alpha)$, while the RC agent generates control signals $u_t$ according to the control policy $\pi_{u}(\tilde{h},\tilde{x}, \alpha; \theta_u)$.
Since we assume plants to have independent dynamics,  we can learn a controller for each plant under approach (2)  to reduce the dimensionality of the learning problem. Then, the AP agent solves the constrained RL problem described in equation \eqref{eq:constrained_optimal_prob_param_single_agent} via a primal-dual algorithm, while the RC agent solves a standard RL problem to minimize the quadratic cost
\begin{equation}
	J(\theta_u) = \mathbb{E}_{x_0}^{\pi_u(\cdot; \theta_u)} \left[ \sum_{t = 0}^{\infty} \gamma^t (x_t^{(i)\trp}Q_i x_t^{(i)} + u_t^{(i)\trp}R_{i} u_t^{(i)}) \right]
\end{equation}
for each plant $i$, with $Q_{i}$ and $R_{i}$ submatrices of $Q$ and $R$ of appropriate dimensions. 

Here, we consider a collection of $m = 10$ linear plants,
\begin{equation}
A^{(i)} = \begin{bmatrix}
-1.01 & 0.5 & 0.5 \\ -0.5 & 1.01 & 0.5 \\ 0 & 0.5 & -0.5 
\end{bmatrix}; \quad
B^{(i)} = \mathbb{I},
\end{equation}
and initial states sampled from a standard normal distribution. The plants are unstable without a control input, with eigenvalues of $A^{(i)}$ closer to 1 than in \eqref{eq:matrices_value_power_allocation}. This means that the plants states do not grow as fast as in \eqref{eq:matrices_value_power_allocation}, making it easier to learn a control policy under long simulation horizons, since the quadratic cost in \eqref{eq:constrained_optimal_prob} will not grow as fast either. Plants are randomly placed in a $[-m/3, m/3]^2$ area around the RC, with $\lambda_h = 1$, $p_l = 2$ and  $W^{(o)} = \mathbb{I}$, and share an overall power budget $p_{\max} = 25m$ (see Example \ref{ex_power_allocation}). 
Performance of the allocation and control policies is measured by a standard LQR cost with $Q = \mathbb{I}$ and $R = 10 ^{-3} \times \mathbb{I}$ over a training horizon $T = 100$. 
Here, $\beta_\theta = \beta_\eta = 5 \times 10^{-4}, \beta_\lambda = 1 \times 10^{-4}$, and $N = 16$. 

 Figure \ref{fig:linear_system_codesign_sims} shows the mean and one standard deviation of (left) the Lagrangian and (center) the constraint violation for the approaches mentioned above. 
 The results show that separating the design of control and allocation policies, with the subsequent reduction in the dimensionality of the learning space, indeed improves performance. The second approach, which relies on iteratively learning control-aware resource allocation policies and channel-aware controllers, achieved a better optimization objective than the approach using a single agent.

The previous experiment shows that separating the control and allocation agents, with the resulting reduction of the dimensionality of the learning problem, improves performance. Next, we compare that approach to potential solutions to the co-design problem, namely 
\begin{enumerate}
\item learning only the control policy (under ideal communication conditions) while using a control-agnostic resource allocation heuristic --- namely, dividing power equally among the plants --- during the test phase;
\item learning only the allocation policy while the plants are controlled by a channel-agnostic, model-based LQR controller.
\end{enumerate}
Results are summarized in Figure \ref{fig:linear_system_codesign_sims} (right).
Each test consists of ten realizations starting from different initial states, and each test point shows the mean and standard deviation for that group of realizations.
The simulation results show that the co-design policy outperforms the learned allocation policy that utilizes a model-based LQR controller in this scenario. Learning only the control policy while dividing power equally among the plants performs worse than both the other approaches in this scenario. Note that this simulation considers linear plants, for which a LQR controller is optimal under ideal communications, so it is expected that a LQR controller in combination with a learned allocation policy would work very well. 
In settings where plant models are known, one could then use a LQR controller as an initialization to the data-driven control policy, or use a model-based but communication-agnostic LQR controller with a learned correction factor that would take communication constraints into account. That might be especially useful as we consider more unstable plants, which makes the control design part of the problem, in particular, more challenging. Nonetheless, we emphasize that without models we are able to jointly learn controller and allocation policy that achieves the performance of model-based LQR in this scenario.

\paragraph{Co-Design --- Nonlinear System}
Next we consider a collection of nonlinear control plants, with each plant corresponding to a classical inverted pendulum model. 
The plants are randomly placed in a $[-m/3, m/3]^2$ area around the RC, with $\lambda_h = 2$, $p_l = 2$ and  $W^{(o)} = 0.1 \times \mathbb{I}$. The objective here is to keep the pendulum upright and the cart around the origin by applying some force $F$.  
The state of each plant is given by $x_t^{(i)} = \left[y^{(i)}, \dot{y}^{(i)}, \theta_p^{(i)}, \dot{\theta_p}^{(i)} \right]$ with $y$ the cart position, $y \in [-4.5, 4.5]$, and $\theta_p$ the pole angle, $\theta_p \in [-12, 12] $ degrees.
Here we consider a system with $m = 10$ plants sharing an overall resource budget $p_{\max} = 25m$ as in Example \ref{ex_power_allocation}. The control input $F$ is constrained to the interval $[-10, 10] $. 
%%%%%%%%%%%%%%%%%%%%%%%%%%%
%%%%%%%%%% FIGURE %%%%%%%%%
%%%%%%%%%%%%%%%%%%%%%%%%%%%
\begin{figure*}[htb]
    \centering
    \begin{minipage}{0.33\textwidth}
        \centering
       \includegraphics[width=\linewidth]{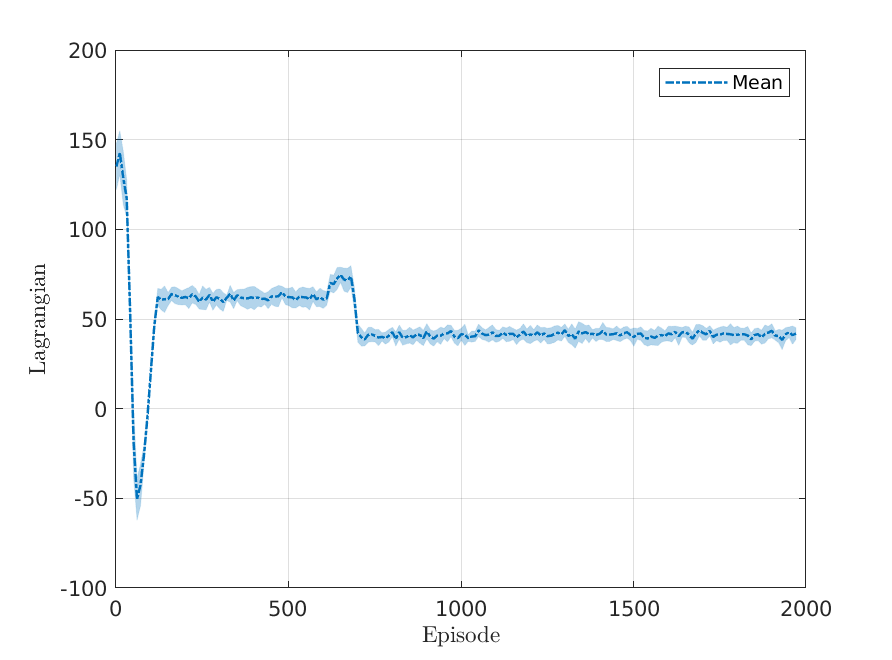}
    \end{minipage} 
        \hfill %
        \begin{minipage}{0.33\textwidth}
        \centering
      \includegraphics[width=\linewidth]{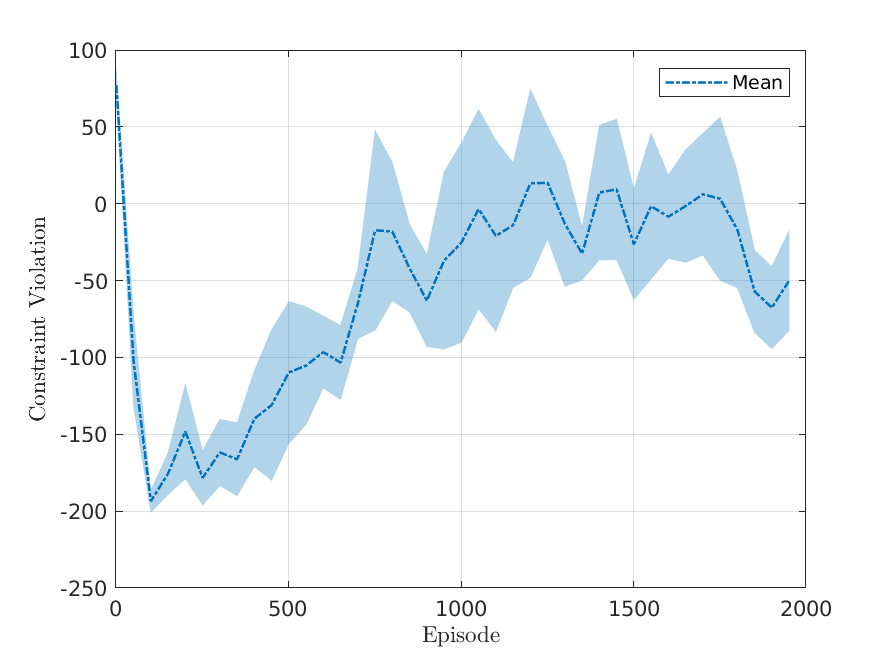}
    \end{minipage}% 
    \hfill %
        \begin{minipage}{0.33\textwidth}
        \centering
      \includegraphics[width=\linewidth]{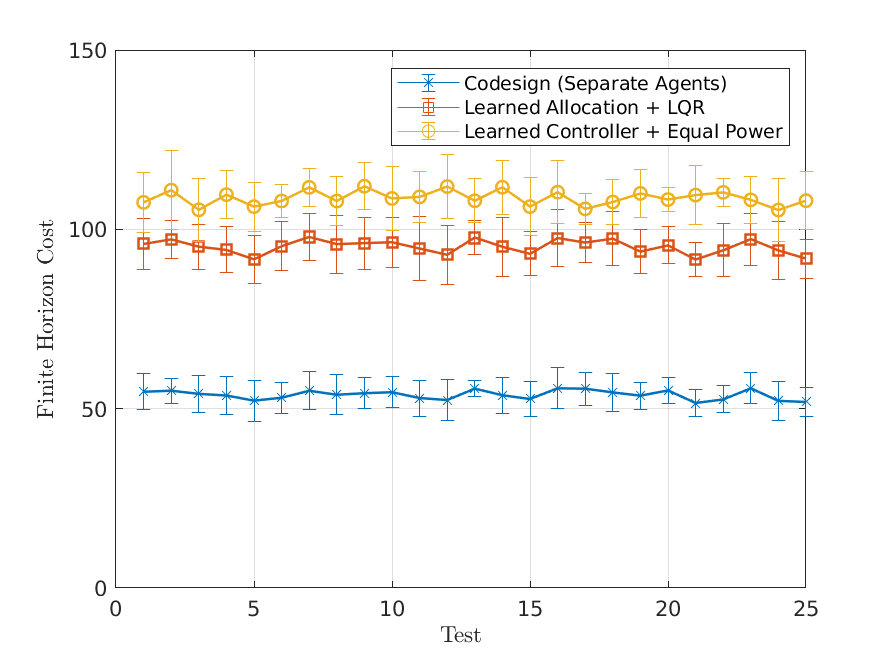}
    \end{minipage}% 
    \caption{Nonlinear system: (left) Lagrangian  and  (center) constraint violation per training episode. (Right) Cost comparison between the model-free codesign policy (blue), learning a resource allocation policy while using a LQR controller based on a local linearization (orange), and learning a control policy while distributing resource equally among the plants (yellow). }
    \label{fig:nonlinear_system_sims}
\end{figure*}
%%%%%%%%%%%%%%%%%%%%%%%%%%%
%%%%%%%%%%%%%%%%%%%%%%%%%%%
%%%%%%%%%%%%%%%%%%%%%%%%%%%

We customize OpenAI Gym's \emph{Cartpole} environment to make the control decisions continuous in $[-10, 10]$ and to account for a linear quadratic cost with
\begin{equation}
Q = 
\begin{bmatrix}
0.1 & 0 & 0 & 0 \\
0 & 0 & 0 & 0 \\
0 & 0 & 1 & 0 \\
0 & 0 & 0 & 0
\end{bmatrix}; \;
R = 0.001 \times \mathbb{I}.
\end{equation}
   The simulation uses a standard Euler integration method with time step $dt = 0.02$. Initial states of the plants are sampled from a uniform distribution, $x_0^{(i)} \sim \mathcal{U}[-0.05, 0.05] $. Policies were trained with $N = 16$ simultaneous realizations over $ 2000 $ episodes of duration $T = 80$, and with  learning rates $\beta_\theta = \beta_\eta = 5 \times 10^{-4}$, $\beta_\lambda = 1 \times 10^{-4}$. To facilitate training, we first trained the codesign controller while dividing power equally among the plants. The codesign policy quickly converges to a feasible policy, as shown in Figure \ref{fig:nonlinear_system_sims} (center).

As in Section \ref{sec:num_exp} (b), we compare the codesigned policy against an allocation policy using a model-based controller and a controller learned under ideal communication conditions. The LQR controller is based on a linearization of 
the nonlinear model around $\theta_p = 0$ and is designed to minimize a quadratic  performance criterion with the matrices $Q$ and $R$ as defined above. To compare the different learned policies, we used a horizon $T = 100$.
Simulation results are shown in Figure \ref{fig:nonlinear_system_sims}.
  Each test point corresponds to an average over 10 realizations starting from different initial states. In this case, the model-free codesign solution (blue) outperforms the solution where the control policy was learned  from scratch while distributing resources equally among the plants. The codesign solution also outperforms the solution that used a standard model-based LQR controller but learned the allocation policy. We point out that in this case the LQR controller is based on a local approximation of the plants valid around the equilibrium point, but nonetheless requires a valid model for the local approximation.  

%%%%%%%%%%%%%%%%%%%%%%%%%%%%%%%
%%% SECTION : Conclusion  %%%
%%%%%%%%%%%%%%%%%%%%%%%%%%%%%%%
\section{Conclusion}

This paper presents a model-free approach to design resource allocation and control policies in wireless control systems subject to long-term constraints on utilization of communication resources or performance of individual plants. 
 We show that the optimality loss resulting from the use of policy parametrizations can be bounded, and present conditions under which near-optimal policies can render the system stable. We also present a practical primal-dual RL algorithm to design control-aware communication policies and communication-aware control policies.
 Numerical experiments attest to the strong performance of the proposed solution over heuristic baselines. Extensions of the proposed approach to other types of access mechanisms, such as URLLC (ultra-reliable low-latency communication) / 5G, is a very interesting future direction in the study of model-free communication and control co-design. Another potential future direction is the extension of the proposed approach to design \emph{uplink} resource allocation policies that minimize the impact of dropouts in the sensing channels on the performance of the control and allocation policies.
Model-based optimal solutions for linear control problems (under ideal communication conditions), on the other hand, are well understood, pointing out to the potential combination of model-free and model-based techniques to design communication-aware controllers in settings where plant models are known. This could be done by learning a correction factor for a standard model-based LQR controller that would take wireless fading conditions and allocation decisions into account, for example, or using known, model-based controllers to initialize control policies via imitation learning.

\section*{Acknowledgments}
This work was supported by Intel Science and Technology Center for Wireless Autonomous Systems and ARL DCIST CRA W9111NF-17-2-0181. 

%%%%%%%%%%%%%%%%%%%%%%%%%%%%%%%
%%% SECTION : Appendix  %%%
%%%%%%%%%%%%%%%%%%%%%%%%%%%%%%%
\appendix
\section{Proof of Theorem \ref{theorem:dual_param_bound}}
\label{proof:dual_param_bound}

Theorem 1 is adapted from  \cite[Theorem 2]{paternain2019constrained}, and the proof falls along similar lines. To start, consider  the dual of the non-parametric resource allocation problem \eqref{eq:constrained_optimal_prob},
\begin{equation}
\begin{aligned}
D^* &= \max_{\lambda \geq 0 } d(\lambda) = \max_{\lambda \geq 0 } \min_{\pi} \mathcal{L} (\pi, \lambda), \\
\mathcal{L} (\pi, \lambda) &:= J(\pi)
+ \lambda^\trp \mathbb{E}_{x_{0}}^{\pi(s; \theta)}  \left[  \sum_{t = 0}^\infty \gamma^t l(x_t, h_t, \alpha_t)  \right].
\end{aligned}
\label{eq:lagrangian_nonparam}
\end{equation}

Strong duality holds for the non-parametrized resource allocation problem, that is, the optimal value of the primal \eqref{eq:constrained_optimal_prob} and dual \eqref{eq:lagrangian_nonparam} problems coincide, as shown in Theorem \ref{theo:nonparam_duality}, adapted from \cite{paternain2019constrained}.
\begin{theorem}[Non-parametric duality gap (Theorem 1, \cite{paternain2019constrained})]
	\label{theo:nonparam_duality} Suppose that the one-step cost $r_t(s,a)$ and the constraints $l^{(i)}(s,a)$ are bounded for all $(s,a) \in \mathcal{S} \times \Gamma$, $i = 1, \dots, r$, and that Slater's condition holds for \eqref{eq:constrained_optimal_prob}. Then, strong duality holds for \eqref{eq:constrained_optimal_prob}.
\label{prop:nonparam_duality_gap}
\end{theorem}
\begin{proof}
See \cite[Theorem 1]{paternain2019constrained}.
\end{proof}

The bound tying the solutions of the constrained co-design problem \eqref{eq:constrained_optimal_prob} and the parametrized dual problem \eqref{eq:dual_opt_problem} depends on the following relaxation of the co-design problem:

 \begin{equation}
 \begin{aligned}
    P_\xi^*=\minimize_{\pi(\cdot)} \, &\mathbb{E}^{\pi(\cdot)}_{x_0}  \left[\sum_{t = 0}^{\infty}\gamma^t \left( x_t^\trp Q x_t  + u_t^\trp R u_t \right) \right] \\
  \text{s.t.} \, &\mathbb{E}_{x_0}^{\pi(\cdot)} \left[ \sum_{t = 0}^\infty \gamma^t l^{(i)}(x, h, \alpha) \right] \leq \xi_i, i = 1, \dots, r  \\
  & \alpha(\tilde{h},\tilde{x}) \in \mathcal{A}; \quad  u(\tilde{h},\tilde{x}) \in \mathcal{U}, \\
 \end{aligned}
 \label{eq:primal_perturbed}
\end{equation}
where $\xi_i \in \mathbb{R}$.
We are now ready to prove Theorem \ref{theorem:dual_param_bound}.
\begin{proof}
Since the introduction of the parameter vector in \eqref{eq:constrained_optimal_prob_param_single_agent} reduces the set of minimizers of the Lagrangian, we have that $d_\theta(\lambda) \geq d(\lambda)$ for any $\lambda \in \mathbb{R}^r$ --- including the solution of the dual problem \eqref{eq:lagrangian_nonparam}, $\lambda^*$ --- and thus
\begin{equation}
D^* = d(\lambda^*) \leq d_\theta(\lambda^*) \leq D_\theta^*.
\end{equation}
From Theorem \ref{prop:nonparam_duality_gap}, on the other hand, we know that the duality gap for the non-parametric constrained problem is null, that is, $P^* = D^*$, from which the lower bound $D_\theta^* \geq P^*$ follows. To find the upper bound in Theorem \ref{theorem:dual_param_bound}, we first rewrite the parametrized dual function \eqref{eq:dual_opt_problem} as
\begin{equation}
d_\theta(\lambda) = d(\lambda) + \left(\min_{\theta} \mathcal{L}(\theta, \lambda) - \min_{\pi \in \mathcal{P}(\mathcal{S})}  \mathcal{L}(\pi, \lambda)  \right).
\end{equation}
Assume $\pi^*$ minimizes the Lagrangian in \eqref{eq:lagrangian_nonparam}, and let $\theta^*$ an $\epsilon$-approximation of $\pi^*$. Then, we have that
\begin{equation}
d_\theta(\lambda) \leq d(\lambda) + \left( \mathcal{L}(\theta^*, \lambda) -   \mathcal{L}(\pi^*, \lambda)  \right).
\label{eq:proof_dual_nonparam_dthetadlamb}
\end{equation}
Now, let us bound the difference above. Note that 
\begin{equation}
 \mathcal{L}(\theta^*, \lambda) -   \mathcal{L}(\pi^*, \lambda)  = \int_{\mathcal{S} \times \Gamma} (c + \lambda^\trp l) (d \rho{_\theta^*} - d \rho{_\pi^*}),
\end{equation}
with $c$ and $l$ standing for the one-step cost and (vector of) constraints, respectively, and $\rho_{\pi^*}$, $\rho_{\theta^*}$ representing the occupation measure induced by $\pi^*$ and $\pi_{\theta^*}$, respectively.  
Since $\theta^*$ is an $\epsilon$-universal approximation of $\pi^*$, we know from Lemma 1 in \cite{paternain2019constrained} that
\begin{equation}
\int_{\mathcal{S} \times \Gamma} | d \rho^*(\lambda) - d \rho^*_\theta(\lambda) | \leq \frac{\epsilon}{1 - \gamma}.
\end{equation}
Thus,
\begin{equation}
\begin{aligned}
\mathcal{L}(\theta^*, \lambda) -   \mathcal{L}(\pi^*, \lambda)  & \leq \left| \int_{\mathcal{S} \times \Gamma} (c + \lambda^\trp l) (d \rho{_\theta^*} - d \rho{_\pi^*}) \right| \\
& \leq \left( B_r + \| \lambda^* \|_1 B_l \right) \frac{\epsilon}{1 - \gamma},
\end{aligned}
\end{equation}
where the last inequality follows from the bounds on the cost and constraint functions stated in Theorem \ref{theorem:dual_param_bound}.
Combining this expression with \eqref{eq:proof_dual_nonparam_dthetadlamb}, we have
\begin{equation}
d_\theta(\lambda)  \leq d (\lambda) + \left( B_r + \| \lambda^* \|_1 B_l \right) \frac{\epsilon}{1 - \gamma}.
\end{equation}
Now let us define 
\begin{equation}
d_\epsilon (\lambda) = d (\lambda) + \left( B_r + \| \lambda^* \|_1 B_l \right) \frac{\epsilon}{1 - \gamma},
\label{eq:d_epsilon_ineq}
\end{equation}
and note that this is the dual function associated to the perturbed primal problem \eqref{eq:primal_perturbed} with perturbation $\xi_i = - B_l \epsilon / (1 - \gamma)$, $i = 1, \dots, m$. Then, we can rewrite \eqref{eq:d_epsilon_ineq} as
\begin{equation}
d_\theta (\lambda) \leq d_\epsilon (\lambda) + B_r  \frac{\epsilon}{1 - \gamma}.
\end{equation}
But that inequality holds for any $\lambda$, and in particular it holds for $\lambda^*_\theta$, the dual solution of the parametrized problem \eqref{eq:dual_opt_problem},
\begin{equation}
D_\theta^* \leq d_\epsilon (\lambda_\theta^*) + B_r  \frac{\epsilon}{1 - \gamma}.
\end{equation}
On the other hand, $\lambda_\epsilon^* = \arg \max d_\epsilon(\lambda)$, and thus
\begin{equation}
D_\theta^* \leq \min_{\pi \in \Pi} J(\pi) + \sum_{i = 1}^r \lambda_{\epsilon, i}^* \left( L_i(\pi) +  B_l  \frac{\epsilon}{1 - \gamma} \right) + B_r  \frac{\epsilon}{1 - \gamma},
\end{equation}
with $L_i(\pi) \coloneqq \mathbb{E}_{x_0}^{\pi(\cdot)} \left[ \sum_{t = 0}^\infty \gamma^t l^{(i)}(x, h, \alpha) \right]$.
But we can upper bound the previous expression by selecting $\pi^*$ the solution of the original constrained co-design problem \eqref{eq:constrained_optimal_prob}, leading to
\begin{equation}
D_\theta^* \leq  J(\pi^*) + \sum_{i = 1}^r \lambda_{\epsilon, i}^* \left( L_i(\pi^*) +  B_l  \frac{\epsilon}{1 - \gamma} \right) + B_r  \frac{\epsilon}{1 - \gamma}.
\end{equation}
Since we assume the problem to be feasible, we have that  $L_i(\pi^*) \leq 0 $, from which we can conclude that
\begin{equation}
D_\theta^* \leq  P^* + \left( B_{r} + \| \lambda_\epsilon^* \|_1 B_l \right) \frac{\epsilon}{1 - \gamma}.
\end{equation}

\end{proof}

	\section{Proof of Theorem \ref{theo:approx_sol_stab}}
	\label{proof:approx_sol_stab}
	
	Let $x_t$ represent the evolution of the plants states under the optimal, non-parametrized policy $\pi^*$, and let $\tilde{x}$ represent the system controlled by a policy $\pi^*_\theta$ parametrized by an $\epsilon$-universal parameter $\theta^*$ found by the dual descent procedure described in Algorithm \ref{alg:dual_descent}. 
	Under the conditions of Theorem \ref{theo:approx_sol_stab}, $x^{(i)} = 0, i = 1, \dots, m$ is an equilibrium point for the wireless control system and $V(0)= 0$, $V_{\theta^*}(0) = 0$. On the other hand, 
	\begin{equation}
		V(x) \geq x^\trp  Q x \geq \varphi(\|x\|) 
	\end{equation}
	with $\varphi(\| x\|) = \lambda_{\min}(Q) \|x\|^2$. The same inequality holds for $V_{\theta^*}(x)$, thus both functions are positive definite \cite{li_stability_2013}. Following \cite{li_stability_2013}, assume that $\nu < r$ and let $\mu$ be the first exit time of $\tilde{x}_t$ from $D_h$, i.e.,
	\begin{equation}
		\mu \coloneqq \inf \left\{ t \geq t_0 \, : \, \tilde{x}_t \not \in D_h \right\}.
	\end{equation}
	We have from \eqref{eq:stab_vf_lyap_cond} that for any $\tilde{x}_t \in D_r \subset \mathbb{R}^{mp}$, 
$\mathbb{E} \left[ V (\tilde{x}_{t + 1})\right] \leq   V (\tilde{x}_{t}) $, and thus (see \cite{li_stability_2013}),
\begin{equation}
\mathbb{E} \left[ V (\tilde{x}_{ \mu \wedge t })\right] \leq   V (x_{0}). 
\end{equation}

	Since we assume dual descent to converge to a regular local minimum, we have from the necessary KKT conditions that either the optimal Lagrange multipliers converge to zero (if the constraints are inactive at that point), or the inequality constraints reduce to zero themselves. Then, Theorem  \ref{theorem:dual_param_bound} implies 
	\begin{equation}
		\mathbb{E} \left[ V_{\theta^*} (\tilde{x}_{\mu \wedge t}) \right] \leq  \mathbb{E} \left[ V (\tilde{x}_{\mu \wedge t}) \right] + \kappa \epsilon \leq V(x_0) + \kappa \epsilon,
	\end{equation}
with $\kappa$ standing for the term multiplying $\epsilon$ in \eqref{eq:dual_param_bound}. Now, if $\mu < t$, then we have that \cite{li_stability_2013}
\begin{equation}
	\varphi(\nu) P \{ \mu \leq t \} \leq \mathbb{E} \left[ \mathbbm{1}_{\mu \leq t} V_\theta(\tilde{x}_\mu) \right] \leq \mathbb{E} \left[ V_\theta(\tilde{x}_\mu) \right].
\end{equation}
But since we assume $V$ and $V_\theta$ to be continuous, we can find, for any $\zeta \in (0, 1)$, $\delta = \delta(\zeta, \nu, t_0) > 0$ such that $V(x) \leq \zeta \varphi(\nu)$ for any $x \in D_\delta$ \cite{li_stability_2013}. Thus, we have that
\begin{equation}
	\varphi(\nu) P \{ \mu \leq t \} \leq V(x_0) + \kappa \epsilon \leq \zeta \varphi(\nu) + \kappa \epsilon
\end{equation}
and thus
\begin{equation}
	P \left\{ \mu \leq t \right\} \leq \zeta + \frac{\kappa \epsilon}{\varphi(\nu)}.
\end{equation}
Finally, letting $t \to \infty$, \eqref{eq:stab_prob} follows.

\bibliography{wl_control}
\end{document}